\documentclass[atmp]{ipart_v1}
\usepackage{geometry} 
\usepackage{graphicx}
\usepackage{amsfonts}
\usepackage{amsthm}
\usepackage{amsmath}
\usepackage{amssymb}
\usepackage{url}
\usepackage[all]{xy}
\usepackage{enumerate}
\usepackage{color}
\usepackage{multirow}
\usepackage{array}
\usepackage[normalem]{ulem}
\usepackage{tikz}
\usepackage{xspace}
\usepackage{enumitem}
\usetikzlibrary{arrows,decorations.markings}
\newtheorem{theorem}{Theorem}
\newtheorem{corollary}{Corollary}
\newtheorem{lemma}{Lemma}
\newtheorem{proposition}{Proposition}
\theoremstyle{definition}
\newtheorem{example}{Example}
\newtheorem{definition}{Definition}

\theoremstyle{remark}
\newtheorem{remark}{Remark}

\newcommand{\bR}{\mathbb R}
\newcommand{\bC}{\mathbb C}

\newcommand{\bZ}{\mathbb Z}
\newcommand{\bP}{\mathbb P}

\newcommand{\Spin}{\mathop{\rm Spin}}

\newcommand{\gen}[1]{\langle #1\rangle}

\newcommand{\cls}[1]{\overline{#1}}
\newcommand{\ph}{\varphi}
\newcommand{\ins}{\subseteq}

\DeclareMathOperator{\Gal}{Gal}

\tikzset{
  big arrow/.style={
    decoration={markings,mark=at position 1 with {\arrow[scale=2.5]{>}}},
    postaction={decorate},
    shorten >=0.4pt}}
\begin{document}
\title[]{Geometrization of $N$-extended 1-Dimensional Supersymmetry Algebras, I.}
\author[C. Doran, J. Kostiuk, K. Iga, G. Landweber, S. M\'{e}ndez-Diez]{Charles Doran, Kevin Iga, Jordan Kostiuk, Greg Landweber, Stefan M\'{e}ndez-Diez}

\begin{abstract}
The problem of classifying off-shell representations of the $N$-extended one-dimensional super Poincar\'{e} algebra is closely related to the study of a class of decorated $N$-regular, $N$-edge colored bipartite graphs known as {\em Adinkras}. In this paper we {\em canonically} realize these graphs as Grothendieck ``dessins d'enfants,'' or Belyi curves uniformized by certain normal torsion-free subgroups of the $(N,N,2)$-triangle group. We exhibit an explicit algebraic model over $\mathbb{Q}(\zeta_{2N})$, as a complete intersection of quadrics in projective space, and use Galois descent to prove that the curves are, in fact, definable over $\mathbb{Q}$ itself. The stage is thereby set for the geometric interpretation of the remaining Adinkra decorations in Part II.
\end{abstract}

\maketitle

\tableofcontents

\section{Introduction}

In mathematics, the term \emph{supersymmetry} is used to describe algebraic structures which possess a $\mathbb{Z}/2\mathbb{Z}$-grading and obey standard sign conventions related to that grading. These algebraic structures can be attached to other mathematical objects which are, say, topological or geometric in nature. As a result, many standard mathematical objects have well-studied ``super" variants, e.g., manifolds $\rightarrow$ super manifolds or Riemann surfaces $\rightarrow$ super Riemann surfaces.

In physics, \textit{supersymmetry} has a much more specific meaning, referring to structures which are equivariant with respect to extensions of the super Poincar\'{e} algebra.
The Lorentz group is the Lie group of isometries of Minkowski space, or more precisely its double cover, replacing ${\rm SO}(1, d-1)$ with ${\rm Spin}(1, d-1)$. The Poincar\'{e} group is the Lorentz group together with translations, $\Spin{(1,d)}\times\bR^{1,d-1}$. The super Poincar\'{e} group is the Lie supergroup obtained by extending the Poincar\'{e} group by infinitesimal odd elements, called supersymmetry generators, whose squares are spacetime derivatives, the infinitesimal generators of translations. At the Lie algebra level, the supersymmetry generators span the odd component of the super Poincar\'{e} algebra. While supersymmetry algebras can refer to extensions of the super Poincar\'{e} algebra, here we will be dealing with only the super Poincar\'{e} algebra.

The physical representations of the super Poincar\'{e} group and super Poincar\'{e} algebra come in two forms. Both are representations on spaces of fields, i.e., maps from Minkowski space to a finite-dimensional $\mathbb{Z}/2\mathbb{Z}$-graded representation of ${\rm Spin}(1, d-1)$. The $\mathbb{Z}/2\mathbb{Z}$-grading decomposes the fields into {\em bosons} and {\em fermions}, and the Lorentz action decomposes the fields into irreducible components, each corresponding to a different type of particle. The assembly of several such particles into a representation of supersymmetry is called a {\em supermultiplet}. The Poincar\'{e} group acts naturally on such spaces of fields, and the question which remains is how the supersymmetry generators in the super Poincar\'{e} algebra will act.

In off-shell representations, the super Poincar\'{e} algebra acts on dynamically unconstrained spaces of fields, while on-shell representations restrict the action to fields which satisfy the equations of motion, usually coming from a Lagrangian via the Euler-Lagrange equations. Although on-shell representations are more complicated physically, they are more natural from the point of view of representation theory. On the other hand, in off-shell representations the supersymmetry is manifest from the description of the particles in the supermultiplet, allowing us to separate the representation theory from the physics, i.e., the Lagrangian, and facilitating quantization.

Graphs known as \emph{Adinkras} were proposed by Faux and Gates in \cite{Faux:2004} as a fruitful way to investigate off-shell representations of the super Poincar\'{e} algebra. These combinatorial objects were rigorously defined, and their connections to Clifford algebras and coding theory explored, in a long series of works by the DFGHILM collaboration \cite{Doran:2007,Doran:2011,Gates:2011}. Adinkras are graphs whose vertices represent the particles in a supermultiplet and whose edges correspond to the supersymmetry generators. In combinatorial terms, Adinkras are $N$-regular, edge $N$-colored bipartite graphs with signs assigned to the edges and heights assigned to the vertices, subject to certain conditions. Details can be found in Section 2 below.

It is useful to think of an \textit{Adinkra} as consisting of a {\em chromotopology}, which captures the underlying bipartite graph with its $N$-coloring, together with two more compatible structures: an {\em odd dashing}, which marks each edge with a sign, and a {\em height assignment}, which labels each of the vertices with an integer. A complete characterization of chromotopologies was achieved in \cite{Doran:2011}. For each $N$, there is a natural chromotopology on the Hamming cube $[0,1]^N$, with vertices labeled by elements of ${\mathbb{F}}_2^N$. The one-skeleton of the Hamming cube serves as a ``universal cover'' for arbitrary chromotopologies, the covering map being realized by taking cosets with respect to doubly even binary linear error correcting codes $C \subseteq {\mathbb{F}}_2^N$.

The purpose of this paper is to show how to canonically associate a Riemann surface to a given chromotopology. The $N$-regular, edge $N$-coloring gives us a cyclic ordering of the edges at each vertex of the graph based on their color. We call such an ordering a \textit{rainbow}. As described in Section $3$, this provides the structure of a ribbon graph, and, following Grothendieck, we are led through the Belyi curve construction to a presentation of the associated Riemann surface as a covering space over $\mathbb{P}^1(\mathbb{C})$, branched over $\{0, 1, \infty\}$, with the graph embedded as the inverse image of the line segment $[0,1]$. One consequence of arriving at a Riemann surface from a chromotopology in this way is that the $2$-faces of this surface are precisely those bounded by $4$-edge cycles with edges colored by two adjacent colors from the rainbow (opposite edges having the same color). Although a change in the order of the colors in our rainbow will a priori yield a different Riemann surface, we show that in fact it results in a global conjugation of the monodromy group of the covering space, and hence in isomorphic Riemann surfaces. This fits well with the expectation from physics, as equivalence under permutation of the colors is a consequence of $R$-symmetry. We also reinterpret a purely combinatorial operation--- the exterior tensor product of Adinkras--- in geometric terms as a multi-point connected sum of the associated Riemann surfaces.

In Section $4$ we go on to give an algebraic presentation of the Riemann surface associated to a chromotopology. This leads us to a canonical Fuchsian uniformization of the curves by normal torsion-free subgroups of the $(N,N,2)$-triangle group. The observation that the $2$-faces are $4$-edge cycles drawn from adjacent rainbow colors is reinterpretted as saying that the uniformization factors through a particular orbicurve corresponding to an index $N$ subgroup of the triangle group. Using this, and a Galois descent, we show further that the Riemann surfaces associated to Adinkra chromotopologies are very special points in moduli, even among Belyi curves. A priori, Belyi curves are defined over some number field; those associated to Adinkras are defined over $\mathbb{Q}$.

Constructions involving Belyi curves have played an important role in several areas of supersymmetric physics in recent years. These include gauge-theoretic applications of dimer models (aka brane tilings) \cite{Hanany}, bipartite field theories and scattering amplitudes \cite{Franco}, and gauge-string duality \cite{Gopakumar}, to name just a few. Although dimer models will play an important role in part II of this paper \cite{Doran:2015}, specifically through the application of the work of Cimasoni and Reshetikhin \cite{Cimasoni:2007}, we make no direct connection here between the geometrization of $N$-extended supersymmetry algebras and these other appearances of Belyi curves in the recent literature.

While not the focus of this paper, we note that the chromotopology ignores two additional structures an Adinkra possesses: an odd dashing and height assignment. In a subsequent paper, \cite{Doran:2015}, we show that the odd dashing defines a spin structure on the associated Riemann surface, which allows us to define a canonical super Riemann surface structure with Ramond punctures following work by Donagi and Witten \cite{Donagi:2013}. We also show in \cite{Doran:2015} that the Adinkra height assignments define a discrete Morse function on the super Riemann surface in the sense of both Banchoff \cite{Banchoff:1970} and Forman \cite{Forman:1998a,Forman:1998b}. The height assignment can also be viewed as a divisor on the (super) Riemann surface. The purpose of this current paper is to give a complete and thorough description of the Riemann surfaces associated to Adinkra chromotopologies, leaving the additional structure these surfaces have to part II.

\section{Review of Adinkras}
\label{sec:RevAd}

As noted in the introduction, the first step in describing irreducible off-shell representations of the $N$-extended $1$-dimensional super Poincar\'{e} algebra is to present them as graphs called Adinkras. We are interested in the elementary $N$-extended Poincar\'{e} superalgebra in $1$-dimensional Minkowski space, also known as the $(1|N)$ superalgebra. Here \emph{elementary} means a classical Lie algebra with no central extensions and no other additional internal bosonic symmetries. We begin this section by reviewing the $(1|N)$ superalgebra. We will then review what an Adinkra is, as well as outlining the main features that will be needed later.

In $1$-dimensional Minkowski space, there is a single time-like direction $\tau$. Translations in this direction are generated by $\partial_\tau$. Therefore $(1|N)$ superalgebras are generated by $\partial_\tau$ and $N$ real supersymmetry generators $Q_I$. The supersymmetry generators commute with $\partial_\tau$ and satisfy the anticommutation relations
\begin{equation}
\label{eq:superalg}
\{Q_I,Q_J\}=2i\delta_{IJ}\partial_\tau,
\end{equation}
where $\delta_{IJ}$ is the Kronecker delta.

In the physics literature, this relation is often written in terms of parameter-dependent operators
\begin{equation}
\delta_Q(\epsilon)\equiv-i\epsilon^IQ_I,
\end{equation}
where $\epsilon^I$ is a set of $N$ Grassmann variables and Einstein's summation convention is being used. With this identification, equation \eqref{eq:superalg} takes the equivalent form
\begin{equation}
[\delta_Q(\epsilon_1),\delta_Q(\epsilon_2)]=2i\epsilon_1^I\epsilon_2^I\partial_\tau.
\end{equation}

Every representation of the $(1|N)$ superalgebra decomposes as a collection of irreducible representations of the $(1|1)$ superalgebra. The $(1|1)$ superalgebra has two irreducible representations, the scalar and spinor multiplets. The scalar multiplet consists of a real commuting bosonic field $\phi$ and a real anticommuting fermionic field $\psi$ with supersymmetry transformations
\begin{align}
Q\,\phi &= \pm\psi,\nonumber\\
Q\,\psi &=\pm i\dot{\phi}.
\label{eq:scaltran}
\end{align}
The spinor representation also consists of a real commuting field $B$ and a real anticommuting field $\eta$, but  with different transformation rules:
\begin{align}
Q\,\eta &=\pm iB,\nonumber\\
Q\, B &=\pm \dot{\eta}.
\label{eq:spintran}
\end{align}

Real, finite-dimensional linear representations of the $(1|N)$ superalgebra are spanned by a basis of real bosonic component fields $\phi_1(\tau),\ldots,\phi_m(\tau)$ and real fermionic component fields $\psi_1(\tau),\ldots,\psi_l(\tau)$. The super supersymmetry generators $Q_1,\ldots,Q_N$ act linearly on the representation and satisfy equation \eqref{eq:superalg}. Such representations are called real supermultiplets, which is why we referred to the scalar and spinor representations of the $(1|1)$ superalgebra as the real and spinor multiplets earlier. We will be interested in off-shell supermultiplets, i.e., supermultiplets whose fields do not satisfy any differential equations other than equation \eqref{eq:superalg}. We will assume all supermultiplets are off-shell unless otherwise stated. An off-shell supermultiplet has as many bosonic component fields as it does fermionic ones.

For off-shell supermultiplets, the supersymmetry transformation rules are
\begin{align}
Q_I\,\phi_A(\tau) &= c\partial_\tau^\lambda\psi_B(\tau),\nonumber\\
Q_I\,\psi_B(\tau) &= \frac{i}{c}\partial_\tau^{1-\lambda}\phi_A(\tau),
\label{eq:suptran}
\end{align}
where $c=\pm1$ and $\lambda=0$ or $1$.

Note that the time derivative has engineering dimension $[\partial_\tau]=1$. It can be seen from equation \eqref{eq:superalg} that $[Q_I]=\frac{1}{2}$. Note that $c$, $\lambda$, and $B$ occurring in equation \eqref{eq:suptran} generally depend on $A$ and $I$. For example, $c$ clearly differentiates between the $\pm$ options  in equations \eqref{eq:scaltran} and \eqref{eq:spintran}. On the other hand, $\lambda$ differentiates between the scalar and spinor multiplet transformations. In order for the component fields to have definite engineering weight, we must have
\begin{equation}
\lambda=[\phi_A]-[\psi_B]+\frac{1}{2},
\end{equation}
assuming the coefficients of equation \eqref{eq:suptran} are dimensionless.

An Adinkra is a graphical representation of a supermultiplet and its supersymmetry transformations, originally proposed in \cite{Faux:2004}. As noted in the introduction, there is a wealth of literature further studying Adinkras and establishing their precise mathematical formulation; see \cite{Doran:2007,Doran:2011,Gates:2011}, for example. A good overview of their mathematical aspects can be found in \cite{Zhang:2011}. An Adinkra is a special bipartite $N$-regular colored graph. The edges have a dashing and an orientation, which defines a height assignment on the vertices.

Consider a $(1|N)$ supermultiplet $\mathcal{M}$ spanned by component fields $\phi_1,\dots,\phi_m,$ $\psi_1,\ldots,\psi_m$. The supermultiplet $\mathcal{M}$ can be represented as an Adinkra if all of the supersymmetry generators send each component field to a single component field. The corresponding Adinkra has a white vertex for each bosonic field $\phi_A$ and a black vertex for each fermionic field $\psi_A$, $1\leq A\leq m$. The white vertex corresponding to $\phi_A$ is connected to the black vertex corresponding to $\psi_B$ by an edge of color $I$ if $Q_I$ sends $\phi_A$ to $\psi_B$ (or its time derivative) by equation \ref{eq:suptran}. The edge is oriented from the white vertex to the black vertex if $\lambda=0$ and the other way if $\lambda=1$. It is dashed if $c=-1$ and solid if $c=1$. This correspondence is depicted in Table \ref{Table:AdinkraAlgebraCorespond}. 
\begin{table}
\centering
\begin{tabular}{c c | c c}
Action of $Q_I$ & Adinkra & Action of $Q_I$ & Adinkra\\
\hline
$Q_I\left[\begin{array}{c} \psi_B\\
\phi_A\end{array}\right]=\left[\begin{array}{c} i\dot{\phi}_A\\
\psi_B\end{array}\right]$ & \begin{tikzpicture}[baseline]
\draw [big arrow] (0,-0.4375) -- (0,0.4375);
\draw[fill] (0,0.5) circle[radius=0.0625];
\draw (0,-0.5) circle[radius=0.0625];
\node[right] at (0,0.5) {\tiny{B}};
\node[right] at (0,-0.5) {\tiny{A}};
\node[left] at (0,0) {\tiny{I}};
\end{tikzpicture} & $Q_I\left[\begin{array}{c} \psi_B\\
\phi_A\end{array}\right]=\left[\begin{array}{c} -i\dot{\phi}_A\\
-\psi_B\end{array}\right]$ & \begin{tikzpicture}[baseline]
\draw[dashed,big arrow] (0,-0.4375) -- (0,0.4375);
\draw[fill] (0,0.5) circle[radius=0.0625];
\draw (0,-0.5) circle[radius=0.0625];
\node[right] at (0,0.5) {\tiny{B}};
\node[right] at (0,-0.5) {\tiny{A}};
\node[left] at (0,0) {\tiny{I}};
\end{tikzpicture}\\
\hline
$Q_I\left[\begin{array}{c} \phi_A\\
\psi_B\end{array}\right]=\left[\begin{array}{c} i\dot{\psi}_B\\
\phi_A\end{array}\right]$ & \begin{tikzpicture}[baseline]
\draw[big arrow] (0,-0.4375) -- (0,0.4375);
\draw (0,0.5) circle[radius=0.0625];
\draw[fill] (0,-0.5) circle[radius=0.0625];
\node[right] at (0,0.5) {\tiny{A}};
\node[right] at (0,-0.5) {\tiny{B}};
\node[left] at (0,0) {\tiny{I}};
\end{tikzpicture} & $Q_I\left[\begin{array}{c} \phi_A\\
\psi_B\end{array}\right]=\left[\begin{array}{c} -i\dot{\psi}_B\\
-\phi_A\end{array}\right]$ & \begin{tikzpicture}[baseline]
\draw[dashed,big arrow] (0,-0.4375) -- (0,0.4375);
\draw (0,0.5) circle[radius=0.0625];
\draw[fill] (0,-0.5) circle[radius=0.0625];
\node[right] at (0,0.5) {\tiny{A}};
\node[right] at (0,-0.5) {\tiny{B}};
\node[left] at (0,0) {\tiny{I}};
\end{tikzpicture}\\
\end{tabular}
\caption{The correspondence between Adinkras and the action of the supersymmetry generators on the component fields. Each white vertex of an Adinkra corresponds to a bosonic component field and its time derivatives. Similarly, each black vertex corresponds to a fermionic component field and its time derivatives. The edges will be colored by color $I$ corresponding to the index of the supersymmetry generator.}
\label{Table:AdinkraAlgebraCorespond}
\end{table}

We  now review some important features of Adinkras that we will need later.
First note that every vertex has exactly one edge of each color adjacent to it. This is because each supercharge acts on each component field, taking it to exactly one other component field (or its time derivative). For any component field $f(\tau)$, $\pm f(\tau)$ and all of its time derivatives are represented by the same vertex in the corresponding Adinkra. By equation \ref{eq:superalg}, $Q_IQ_J=-Q_JQ_I$ for $I\neq J$, so $Q_IQ_Jf$ and $-Q_JQ_If$ are represented by the same vertex in an Adinkra. Thus, traveling along an edge of color $I$ and then an edge of color $J \not= I$ is the same as traveling first along an edge of color $J$ and then along one of color $I$. Another way to say this is that starting at any vertex and taking an edge of color $I$, then one of color $J \not= I$, then $I$ again, and finally $J$ returns us to the same vertex \cite{Faux:2004}. We refer to such a closed loop as a \textit{$2$-colored loop}.

As we travel around each $2$-colored loop, there must be an odd number of dashed edges \cite{Faux:2004}. This follows from the anticommutativity of the supersymmetry generators and the fact that the dashedness of an edge corresponds to the sign of the supersymmetry tranformation; see Table \ref{Table:AdinkraAlgebraCorespond}.

Furthermore, as we travel along a $2$-colored loop counter-clockwise, we  will travel along an even number of edges in agreement with their orientations and and even number of edges against their orientations. Indeed, if we start at a fixed vertex and travel along an edge of color $I$ and then along an edge of color $J$, whether or not you go with the orientations of the edges determines how many factors of $\partial_\tau$ are picked up. Since we must have the same number of time derivatives if we had instead went clockwise (for the engineering dimension to match), we must have travelled with or against the orientations of the same number of edges. Therefore the number of edges we travelled that agreed with the orientations, and the number of edges we travelled against the orientations as we travel along the entire $2$-colored loop must both be even.

The orientation of the edges defines a height function on the vertices of the Adinkra that corresponds to the engineering dimension of the component fields \cite{Doran:2007}. If $V$ is the set of vertices of an Adinkra, a height function is a function $h:V\to\bZ$ such that $h(b)=h(a)+1$ if there is an edge going from $a$ to $b$. Adding a constant to any height function gives another height function. This allows us to choose a height function so that the heights of the bosons are even and the heights of the fermions are odd, which can be normalized so that the height of each vertex is twice the engineering dimension of the corresponding component field. The height of a vertex should be viewed as a minimum engineering dimension of the objects represented by the vertex, since the vertices of an Adinkra represent not only the component fields but also their time derivatives, and every application of $\partial_\tau$ increases the engineering dimension by $1$. See Figure \ref{Fig:N2Adinks} for an example of two different $N=2$ Adinkras with their height functions shown.

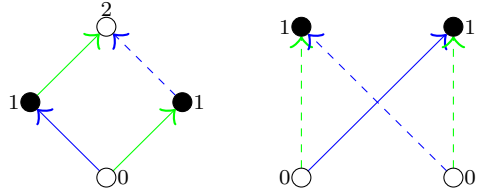
\begin{figure}
\centering
\begin{tabular}{c c c}
\begin{tikzpicture}[baseline,scale=2]
\draw[big arrow,green] (0.0441942,-0.4558058) -- (0.4558058,-0.0441942);
\draw[big arrow,blue] (-0.0441942,-0.4558058) -- (-0.4558058,-0.0441942);
\draw[big arrow,blue,dashed] (0.4558058,0.0441942) -- (0.0441942,0.4558058);
\draw[big arrow,green] (-0.4558058,0.0441942) -- (-0.0441942,0.4558058);
\draw (0,0.5) circle[radius=0.0625];
\draw (0,-0.5) circle[radius=0.0625];
\draw[fill] (0.5,0) circle[radius=0.0625];
\draw[fill] (-0.5,0) circle[radius=0.0625];
\node[right] at(0,-0.5) {\tiny{$0$}};
\node[above] at(0,0.5) {\tiny{$2$}};
\node[right] at(0.5,0) {\tiny{$1$}};
\node[left] at(-0.5,0) {\tiny{$1$}};
\end{tikzpicture} & & \begin{tikzpicture}[baseline, scale=2]
 \draw[big arrow,green,dashed] (-0.5,-0.4375) -- (-0.5,0.4375);
\draw[big arrow,blue] (-0.4558058,-0.4558058) -- (0.4558058,0.4558058);
\draw[big arrow,blue,dashed] (0.4558058,-0.4558058) -- (-0.4558058,0.4558058);
\draw[big arrow,green,dashed] (0.5,-0.4375) -- (0.5,0.4375);
\draw (-0.5,-0.5) circle[radius=0.0625];
\draw (0.5,-0.5) circle[radius=0.0625];
\draw[fill] (0.5,0.5) circle[radius=0.0625];
\draw[fill] (-0.5,0.5) circle[radius=0.0625];
\node[left] at(-0.5,-0.5) {\tiny{$0$}};
\node[right] at(0.5,0.5) {\tiny{$1$}};
\node[left] at(-0.5,0.5) {\tiny{$1$}};
\node[right] at(0.5,-0.5) {\tiny{$0$}};
\end{tikzpicture}
\end{tabular}
\caption{An example of two different $N=2$ Adinkras with the heights of their vertices labeled.}
\label{Fig:N2Adinks}
\end{figure}

Forgetting the dashing, orientation, and coloring of all of the edges in an Adinkra leaves what is called the \textit{topology} of the Adinkra. The topology of an Adinkra together with its edge coloring is called its \textit{chromotopology}. One of the most important Adinkra topologies is the topology of the $N$-cube, $[0,1]^N$. This is the Adinkra topology consisting of the vertices and edges of the $N$-cube. The two Adinkra topologies depicted in Figure \ref{Fig:N2Adinks} are those of the $2$-cube. In fact, they have the same chromotopologies. More generally, every $N$-cube has a unique chromotopology \cite{Doran:2011}. Even though the Adinkras in Figure \ref{Fig:N2Adinks} have the same chromotopologies, they are not the same Adinkras, since they have different orientations (hence height assignments) and dashings. For the remainder of the paper we will ignore Adinkra height assignments and odd dashings and focus on chromotopologies.

The colored $N$-cube, $[0,1]^N$, has $2^N$ vertices and $2^{N-1}N$ edges. We can embed the colored $N$-cube in $\bR^N$ so that the vertices are located at all $2^N$ possible points $(x_1,\ldots,x_N)$ with $x_i=0$ or $1$. In this way, we may associate the vertices of the colored $N$-cube with the elements of $\mathbb{F}_2^N$, where $\mathbb{F}_2$ is the field of two elements. The weight of a vertex is the number of nonzero entries in $(x_1,\ldots,x_N)$. The vertices with even weight are declared white, while the vertices with odd weight are declared black. Two vertices $(x_1,\ldots,x_N)$ and $(y_1,\ldots,y_N)$ are connected by an edge of color $I$ if they differ only in the $I$-th component, i.e., if $x_i=y_i$ for $i\neq I$ and $x_I=1-y_I$.

It was shown in \cite{Doran:2011} that the set of Adinkra chromotopologies is equivalent to the set of colored $N$-cubes mod doubly even codes. A \textit{code} is a linear subspace of $\mathbb{F}_2^N$. A code $\mathcal{C}$ is doubly even if every codeword (element of the code) has weight divisible by $4$. Every code has a basis since it is a linear subspace of $\mathbb{F}_2^N$, and any basis for $\mathcal{C}$ is called a generating set.  The dimension of the code is its dimension as $\mathbb{F}_2$-vector space. For a doubly even code $\mathcal{C}$, $\mathbb{F}_2/\mathcal{C}$ means that we identify vertices of the $N$-cube if they differ by codewords as elements of $\mathbb{F}_2^N$. Furthermore, if the vertex $v$ and the vertex $w$ are identified then, for all $I$,  the edge of color $I$ incident to $v$ is identified with the edge of color $I$ incident to $w$. That is, the vertices of an Adinkra can be viewed as cosets of a doubly even code in $\mathbb{F}_2^N$. In this way, the colored $N$-cube can be thought of as a universal cover for general Adinkra chromotopologies. It is in general quite difficult to find doubly even codes for a given $N$; see \cite{Miller:2013} for an extensive list.

An $R$-symmetry is a symmetry that transforms the supercharges $Q_I$. For real $N$-extended supersymmetry, the group of $R$-symmetries is $O(N)$. Permutation matrices in $O(N)$ permute the $Q_I$. We interpret the action of the permutation subgroup as a permutation action on the colors corresponding to the $Q_I$. As noted in \cite{Doran:2011}, it is not clear that for physical significance it is enough to only consider permutation equivalence. In the following section we will show that Adinkra chromotopologies describe Riemann surfaces and that these surfaces are invariant under the action of the permutation subgroup. The same issues of extending to the entire $R$-symmetry group still exist, but working in a higher-dimensional space may provide new methods for approaching the problem.
\section{The Belyi Curve Associated to an Adinkra Chromotopology}
\label{sec:chromRS}

In the previous section we reviewed how irreducible off-shell representations of $(1|N)$ superalgebras can be presented as graphs called Adinkras. In this section we will see how an Adinkra chromotopology canonically defines a Riemann surface as a covering space of $\bP^1(\bC)$. We do this by first showing that an Adinkra chromotopology has the structure of a ribbon graph and then using the Grothendieck correspondence to associate a Riemann surface to the ribbon graph. We will show that all of the surfaces  we will consider factor through a fixed orbifold, allowing us to study these surfaces in a uniform fashion. After explaining the construction of the Riemann surface, via Grothendieck's theory of ``dessins d'enfants'', we will show that if two Adinkras are related by the permutation subgroup of $R$-symmetry then their corresponding Riemann surfaces are equivalent. Lastly, we will use the result on $R$-symmetric Adinkras to describe how the tensor product of Adinkras can be extended to a well-defined operation on the associated surfaces. 


\subsection{From Chromotopologies to Riemann Surfaces}
\begin{definition}
A \textit{ribbon graph} (also known as a \textit{fat graph}) is a connected graph that assigns to each vertex of the graph a cyclic permutation of the half edges adjacent to the vertex. 

A \textit{rainbow} is a choice of a cyclic ordering of the colors of a colored graph.
\end{definition}

As we saw in the previous section, the chromotopology of an Adinkra is a connected, colored, $N$-regular bipartite graph. 
An Adinkra naturally determines a rainbow given by the order of the supersymmetry generators the colors represent. A chromotopology, together with a rainbow, therefore defines a ribbon graph. The ribbon structure is defined as follows: Given any vertex, the rainbow provides a cyclic permutation of the half-edges incident to that vertex, since there is exactly one half-edge of each color adjacent to it. To each white vertex, we assign precisely this cyclic permutation; to each black vertex, we assign the permutation of the half-edges adjacent to it in the opposite order of the rainbow\footnote{We could have let the elements of $S_{2d}$ at each black vertex have the same order as the rainbow instead, as discussed in \cite{Dessins:2009,Jones:1997}, but this would not take into account the bipartite structure.}.

\begin{definition}
A \emph{dessin d'enfant} (or just \emph{dessin}) is  a pair $(X,\mathcal{D})$ where $X$ is an oriented, compact topological surface and $\mathcal{D}$ is a finite, connected, bipartite graph forming the $1$-skeleton of $X$, i.e., $X-\mathcal{D}$ is the union of finitely many topological discs, called the \emph{faces} of $X$. 
\end{definition}

The Grothendieck correspondence states that a ribbon graph is equivalent to a dessin d'enfant. We will briefly explain the equivalence here, but the reader is encouraged to consult \cite{Jones:1997} or \cite{Dessins:2009} for a more rigorous introduction. As described in \cite{Jones:1997}, the Riemann surface $X$ is built by using the ribbon graph as a $1$-skeleton for $X$ and then ``filling in'' $X$ by attaching $2$-cells
 corresponding to certain closed loops in the graph. Which loops we attach $2$-cells to is determined by the ribbon structure as follows. Suppose we work with a rainbow $(C_1,\dots, C_N)$, and fix a white vertex $w_1$ and a color $C_i$. If we leave this vertex along the half-edge of color $C_i$, we will reach a black vertex $b_1$. Since the ordering at black vertices is opposite to that of the rainbow, we will leave along the half-edge of color $C_{i-1}$ and end up at a white vertex $w_2$ that must be different than $w_1$, since an Adinkra does not have any double edges. At this point, we leave $w_2$ along the half-edge of color $C_i$ to end up at a black vertex $b_2\neq b_1$. Finally, we leave along the half-edge of color $C_{i-1}$ and end up back at $w_1$, since we have now completed a $C_{i-1}/C_{i}$-colored loop. Varying the initial white vertex $w_1$ and color $C_i$, we see that we are attaching $2$-cells to every $C_{i}/C_{i+1}$ colored loop. In this way, we obtain a Riemann surface with the Adinkra as its $1$-skeleton. Note that the faces in this case will be $4$-gons. If we had chosen the order of the colors at the white and black vertices to be the same, then the faces would have been $2N$-gons \cite{Jones:1997}.

In the other direction, if we have a dessin $(X,\mathcal{D})$ then $\mathcal{D}$ is a ribbon graph if we take the cyclic ordering of the half-edges at the white vertices to be the order given by moving counter-clockwise around the white vertex relative to the orientation, and we take the opposite order at the black vertices.

In turn, a dessin d'enfant is equivalent to a Belyi pair, which we now define precisely. 

\begin{definition}
A \it{Belyi pair} $(X,\beta)$ is a closed Riemann surface $X$ equipped with a Belyi map, $\beta:X\to\bP^1(\bC)$ that is ramified at most over $\{0,1,\infty\}$. We refer to $\mathbb{P}^1(\mathbb{C})$ as the \emph{Belyi base}. 
\end{definition}

A dessin naturally defines a Belyi map, namely the map that sends the white vertices to $0$, the black vertices to $1$, the edges to the interval $(0,1)$, and each face to $\mathbb{C}\mathbb{P}^1- [0,1]$, with the center of each face being mapped to $\infty$. In the other direction, given a Belyi pair, we obtain a dessin by taking for the embedded graph $\mathcal{D}$ the pre-image of $[0,1]$. The white vertices are given by the fiber over $0$, the black vertices are given by the fiber over $1$, and the edges are given by the pre-images of the open interval $(0,1)$. This shows that every chromotopology with a rainbow  determines a Belyi pair. For more details, the reader is encouraged to consult \cite{Girondo:Text}.

Using the classification of Adinkra chromotopologies, given in \cite{Doran:2011} and discussed in Section \ref{sec:RevAd}, we refer to an an Adinkra chromotopology obtained from quotienting the colored $N$-cube by a $k$-dimensional doubly even code as an $(N,k)$ \emph{Adinkra chromotopology}. We denote the set of all such chromotopolgies by $\mathcal{A}_{(N,k)}$. Note that for a given $N$ and $k$ there may be more than one chromotopology. In particular, two elements of $\mathcal{A}_{(N,k)}$ will not be equivalent as chromotopologies if they are quotients of the $N$-cube by permutationally inequivalent codes \cite{Doran:2011}. 

\begin{definition}
If $A\in\mathcal{A}_{(N,k)}$ is an Adinkra chromotopology, then $X_A$ denotes the Riemann surface built from $A$ as described above. The set of all Riemann surfaces constructed in this manner will be denoted by $\mathcal{X}_{(N,k)}$.
\end{definition}

Let $A\in\mathcal{A}_{(N,k)}$ and $X=X_A$. The Grothendieck correspondence shows that all of the data of the Belyi pair $(X,\beta)$ can be encoded by its monodromy action. This data is completely determined by two elements $\sigma_0,\sigma_1\in S_d$, where $d$ is the degree of $\beta$. To define these elements, one chooses an unbranched value $z$ on the Belyi base $\mathbb{C}\mathbb{P}^1$ and considers the action of the fundamental group $\pi_1(\mathbb{C}-\{0,1\})$ on the fiber $\beta^{-1}(z)$; $\sigma_0$ describes the action of a simple loop around $0$, while $\sigma_1$ describes the action of a simple loop around $1$, both loops being followed counter-clockwise. We will describe the elements $\sigma_0$ and $\sigma_1$ as permutations of the edges of the embedded graph in $X$ since the Belyi map $\beta$ is unramified over the edges.

Since $A$ has $2^{N-k-1}$ white vertices, there are $d=2^{N-k-1}N$ edges in total, as there is a unique edge of a given color coming out of every white vertex. Since $\beta$ is unramified over the edges, it follows that $\beta$ has degree $d$.

Now let us describe $\sigma_0$ and $\sigma_1$.  Assume that we have labeled the edges of $A$ in some way. At each white vertex $w$, define an $N$-cycle $\sigma_0^{(w)}$ by listing the $N$ edges incident to $w$ in the order of the rainbow. Similarly, for each black vertex $b$, let $\sigma_1^{(b)}$ be the $N$-cycle obtained by listing the $N$ edges incident to $b$ in the opposite order of the rainbow. Then we have
$$\sigma_0=\prod_w \sigma_0^{(w)},\ \ 
\sigma_1=\prod_b \sigma_0^{(b)},$$
the products being taken over the $2^{N-k-1}$ white and black vertices, respectively. 

\begin{definition}
The pair of elements $(\sigma_0,\sigma_1)$ constructed above is called the \emph{permutation representation pair} for the Belyi curve $(X_{(N,k)},\beta)$. 
\end{definition}

We will now demonstrate how these elements describe the monodromy action. Let $\Sigma_0\subseteq\mathbb{C}\mathbb{P}^1$ denote the graph consisting of the closed interval $[0,1]$, and let $0$ and $1$ be the white vertex and black vertex respectively. By construction of $\beta$, we have $\beta(A)=\Sigma_0$. Consider a small loop $\gamma$ that travels counter-clockwise around $0\in\bP^1(\bC)$, taking the base point to be where $\gamma$ intersects the single edge $\Sigma_0$. There are $2^{N-k-1}N$ lifts of the base point, each lying on a unique edge. By construction, the lift of $\gamma$ with initial point lying on the edge of color $C_i$ incident to the white vertex $w$ has terminal point lying on the edge of color $C_{i+1}$ incident to the same white vertex. Therefore, the monodromy action at $0$ is described by sending the edge of color $C_i$ incident to $w$ to the edge of color $C_{i+1}$ incident to $w$. Such a lift has order $N$, which is why $\sigma_0$ is the disjoint product of $N$-cycles when presented as an element of $S_d$; a similar argument applies to $\sigma_1$. 
 
Note that the order of the map $\beta$ is $d$. This can be determined from $\sigma_i\in S_d$, $i\in\{0,1\}$. Since $\sigma_i$ contains $2^{N-k-1}$ disjoint $N$-cycles, the order of $\beta$ is the number of disjoint cycles in $\sigma_i$ times the length of the cycles themselves.  Further, the fact that $\sigma_i$ is the product of disjoint $N$-cycles corresponds to $X$ having order $N$ ramification over the Belyi base at each vertex. This could have also been seen directly, since the degree of $\beta$ is $2^{N-k-1}N$, but there are only $2^{N-k-1}$ vertices of each color. 

The monodromy over $\infty$ is given by
\begin{equation}
\sigma_\infty=\sigma_1\sigma_0.
\end{equation}
Note that it is common to see $\sigma_\infty$ defined as $(\sigma_0\sigma_1)^{-1}$. This gives an equivalent description of the faces, as in \cite{Girondo:Text}.
The element $\sigma_\infty$ is the product of $2^{N-k-2}N$ disjoint $2$-cycles, with each transposition consisting of two of the edges of the same color that make up a $2$-colored face. Note that these two edges uniquely determine the face. To see how $\sigma_\infty$ describes the monodromy, we consider a loop $\gamma$ that travels clockwise around $\infty\in\bP^1(\bC)$, relative to the orientation of $\bP^1(\bC)$ with base point lying on the edge $\Sigma_0$. By construction, the lift of $\gamma$ with initial point lying on the edge of color $C_i$ has terminal point lying on the edge of color $C_i$ that makes up the other edge in the $C_i/C_{i+1}$ $2$-colored face. This is exactly the data encoded in $\sigma_\infty$.

 We could also view $\sigma_\infty$ as listing the two edges that are incoming to the two white vertices that make up the face (or the two edges that are outgoing from the two black vertices in the face). Here, \emph{incoming} and \emph{outgoing} refer to movement along a lift of $\gamma$ in the clockwise direction relative to the orientation that $X$ inerhits from $\beta\colon X\to\mathbb{P}^1(\mathbb{C})$. We have chosen ``clockwise'' here so that the notion of ``incoming'' and ``outgoing matches'' the rainbow at the white vertices. If we instead chose ``counter-clockwise'', we would just need to exchange ``outgoing'' and ``incoming''. This is a result of choosing counter-clockwise as the orientation of the rainbow at the white vertices. 
 
The fact that $\sigma_\infty$ is a product of disjoint $2$-cycles corresponds to $X_{(N,k)}$ having order $2$ ramification over $\infty$, which can again be seen directly by noting that there are precisely $2^{N-k-2}N$ faces.

While $\sigma_\infty$ completely determines the $2$-cells of the Belyi pair, it is not always as convenient since the four edges making up each face cannot be determined at a glance. It is therefore sometimes convenient to look at the element $\pi_\infty$ that consists of $2^{N-k-2}N$ disjoint $4$-cycles, each cycle listing the edges that make up each face as we move around the face clockwise relative to the orientation. We obtain $\sigma_\infty$ from $\pi_\infty$ by dropping the edges incoming at the black vertices.

\begin{proposition}
\label{prop:genus}
For $N\geq 2$, the genus of $X\in \mathcal{X}_{(N,k)}$  is $g=1+2^{N-k-3}(N-4)$. Furthermore, for $1\leq N \leq 3$ there are no doubly even codes and any $X\in\mathcal{X}_{(N,0)}$ has genus $0$.
\end{proposition}

\begin{proof}
According to \cite[Prop 4.10]{Girondo:Text}, we have that
\begin{eqnarray}
2-2g&=&(\#\{\textrm{cycles of}\  \sigma_0\}+\#\{\textrm{cycles of} \ \sigma_1\})\nonumber\\
&&-\#\{\textrm{cycles of}\ \sigma_\infty\}.\nonumber
\end{eqnarray}

Since there are $2^{N-k-1}$ cycles in $\sigma_0$ and $\sigma_1$, and since there are $2^{N-k-2}N$ cycles in $\sigma_\infty$, we conclude that
$$2-2g=2^{N-k}-2^{N-k-2}N.$$

Note that this is just the Euler characteristic coming from the cellular decomposition of $X$ as a dessin. Solving for $g$ yields the formula. 

The maximum weight of an element of $\mathbb{F}_2^N$ is $N$. Therefore, for $N\leq 3$ the weight of a codeword cannot be divisible by $4$, so there can be no doubly even codes. For $N=2,3$, the formula for the genus with $k=0$ shows the genus is $0$. For $N=1$, the associated surface is the Belyi base, whose  genus is $0$.
\end{proof}

Since all the information about $X\in\mathcal{X}_{(N,k)}$ is encoded in its monodromy representation, we will generally use the monodromy representation to describe the Belyi pair $(X,\beta)$.


\subsection{Covering Space Theory for Adinkras}
\label{sec:GalCov}

Consider the Belyi pair $(B_N,\tilde{\beta})$, defined in \cite{Jones:1997}, given by $\mathbb{C}\mathbb{P}^1$ with a single white vertex at $0$, a single black vertex at $\infty$, and $N$ edges joining the two points given by lines with argument equal to $\frac{2\pi j}{N}$ for $j=1,\dots, N$. The Belyi map for $B_N$ is given by\footnote{There is a typo in \cite{Jones:1997} that incorrectly states $\tilde{\beta}=\left(\frac{x}{x-1}\right)^n$. This error is carried through to the computation of $\tilde{\beta}^{-1}$.}

$$\tilde{\beta}(x)=\frac{x^N}{x^N+1}.$$

Note that $(B_N,\tilde{\beta})$ is indeed a Belyi pair: $\tilde{\beta}$ is a degree $N$ covering that has order $N$ ramification over $0$ and $1$ (at $0$ and $\infty$ in $B_N)$, and is unramified everywhere else. 
In this section, we will first show that all of the Belyi pairs $(X,\beta)$ for $X\in\mathcal{X}_{(N,k)}$ factor through $B_N$. After establishing this fact, we will show if $X\in\mathcal{X}_{(N,0)}$ corresponds to the hypercube Adinkra $A$ and $X'\in\mathcal{X}_{(N,k)}$ corresponds to the Adinkra obtained by quotienting $A$ by a doubly even code $\mathcal{C}_k$, then the map $X\to B_N$ factors through $X'\to B_N$.

In order to proceed, we will need to describe the monodromy elements $\sigma_0$ and $\sigma_1$ more explicitly. We now fix a doubly even code $\mathcal{C}_k\subseteq C_N$, where $C_N$ is the maximal even code inside  $\mathbb{F}_2^N$, and consider the associated Adinkra $A$ with rainbow $(1,2,\dots,N)$; let $(X,\beta)$ be the Belyi pair associated to $A$. 
As described above, everything is determined once we fix a labeling of the edges. The white vertices of $A_{A}$ are the elements of the orbit space $C/\mathcal{C}_k$, while the black vertices are the elements of $D/\mathcal{C}_k$, where $D$ is the set of odd elements in $\mathbb{F}_2^N$. Each edge of color $i$ is incident to a unique white vertex $c$; let us call this edge $i_c$. Let $I=\{i_c| i=1,\dots, N, c\in C/\mathcal{C}_k\}$. We will describe the monodromies as elements of $S_I$. Since the rainbow is given by $(1,2,\dots, N)$, the way we have labeled the edges in this case makes writing down $\sigma_0$ quite simple. We see at once that
$$\sigma_0^{(c)}=(1_c,\dots,N_c),$$
and therefore 
$$\sigma_0=\prod_{c\in C/\mathcal{C}_k}(1_c,\dots, N_c).$$

We need to do a little more work to describe $\sigma_1$. The edge of color $i$ incident to a black vertex indexed by $d\in D/\mathcal{C}_k$ is incident to the white vertex $d+e_i\in C/\mathcal{C}_k$, where $e_i$ is the $i$-th standard basis vector of $\mathbb{F}_2^N$. Note that throughout we are working with equivalence classes of elements in $\mathbb{F}_2^N/\mathcal{C}_k$. It follows that
$$\sigma_1^{(d)}=(N_{d+e_N},\dots, 1_{d+e_1}),$$
and therefore
$$\sigma_1=\prod_{d\in D/\mathcal{C}_k}(N_{d+e_N},\dots, 1_{d+e_1}).$$

Before computing the product, let us introduce some notation. For each $i=1,\dots, N-1$, let $c_i$ be the element of $C$ that is zero everywhere except for the $i$-th and $(i+1)$-th positions. Note that $\{c_i\}$ is a generating set for $C$. Let $$c_N=\sum_{i=1}^{N-1}c_i.$$ The element $c_N$ has zero entries everywhere except for the first entry and last entry. 

We can compute $\sigma_\infty=\sigma_1\sigma_0$ explicitly from the above description. If we start with an edge $i_c$, then $\sigma_0$ takes $i_c$ to $(i+1)_{c}$ with the convention that we compute $i+1$ modulo $N$ using $\{1,\dots, N\}$ as a set of representatives. Applying $\sigma_1$, we obtain the element $i_{c+c_{i}}$;thus, $\sigma_1\sigma_0$ takes $i_c$ to $i_{c+c_i}$. 
Therefore, if we let $H_i$ be a set of orbit representatives for the action of $\gen{c_i}$ on  $C/\mathcal{C}_k$, we can write
$$\sigma_\infty=\prod_{i=1}^N\prod_{c\in H_i}(i_c,i_{c+c_i}).$$
Choosing a different set of orbit representative amounts to possibly  changing the order of the two elements in each transposition, which does not change the group element. Therefore, it does not matter how we choose such a set of representatives. 
The white vertex $c+c_i$ is obtained from $c$ by traveling along the edge of color $i$ and then the edge of color $i+1$, so that we could have predicted how $\sigma_\infty$ would look from our earlier description. Using this labeling, we remark that
$$\pi_\infty=\prod_{i=1}^N\prod_{c\in H_i}(i_c,(i+1)_c,i_{c+c_i},(i+1)_{c+c_i}).$$

\begin{theorem}
\label{Beachball}
The Belyi pair $(X,\beta)$ with rainbow $(1,2,\ldots,N)$ factors through the Belyi pair $(B_N,\tilde{\beta})$ with rainbow $(1,2,\dots, N)$ with the edge of color $j$ being given by ray having argument $2\pi j/N$.
\end{theorem}

\begin{proof} Jones proved the result in \cite{Jones:1997} for the Riemann surface  associated to the $N$-cube Adinkra, and we will now extend this result the Riemann surfaces in $\mathcal{X}_{(N,k)}$ for any $k$. First, we will recall the proof given in \cite{Jones:1997} for the $N$-cube. Let $A_N$ denote the $N$-cube Adinkra and let $X_N$ denote the associated Riemann surface. 
The Belyi pair for $(X_N,\beta)$ factors through $(B_N,\tilde{\beta})$ because the automorphism group of the $N$-cube contains a normal subgroup isomorphic to the direct sum of $2^{N-1}$ copies of $\mathbb{Z}/2\mathbb{Z}$ generated by half-turns of faces. Note that the generating set equates the two white vertices that make up a face with each other, as well as equating the two black vertices that are incident to the face with each other. 

In the language of \cite{Doran:2011}, the quotient $f_{X_N}\colon X_N\to B_N$ is the quotient of the cube by the maximal even subcode $C_N\ins\mathbb{F}_2^N$.  The quotient of $A_N$ by $C$ is clearly the embedded graph of $B_N$; let us call it $\Sigma_N$. All of the white vertices are in $C_N$, while all of the black vertices lie outside of $C_N$. Following \cite{Doran:2011}, we equate edges of the same color incident to equivalent points. Therefore all of the edges of a given color are equated, so that $f_{X_N}(A_N)=\Sigma_N$. Note that the generators of $C_N$ connect the two white vertices incident to the $2$-colored faces, so the generators of $C_N$ are equivalent to the generators of the subgroup $\textrm{Aut}(A_N)$. 

The desired factorization follows from the compatibility of the monodromy actions. Therefore, we have a commutative diagram
$$\xymatrix{\pi_1(\mathbb{C}-\{0,1\})\times\beta^{-1}(e)\ar[r]\ar_{\mathrm{id}\times f_{X_N}}[d]&\beta^{-1}(e)\ar[d]^{f_{X_N}}\\
\pi_1(\mathbb{C}-\{0,1\})\times\tilde{\beta}^{-1}(e)\ar[r]&\tilde{\beta}^{-1}(e)}$$
where $e$ denotes the single edge in the Belyi base and the horizontal arrows are given by the monodromy actions. 

The monodromy of $\tilde{\beta}$ is given by $\tilde{\sigma}_0=(1,2,\dots,N)$ and $\tilde{\sigma}_1=(N,\dots, 1)$, where we label the edge of color $i$ by $i$. If we start with an edge $i_c$ on the top-left and follow the diagram clockwise, we reach the edge 
$$f_{X_N}(\sigma_0(i_c))=f_{X_N}((i+1)_c)=i+1$$
by the construction of $f_{X_N}$. On the other hand, following the diagram counter-clockwise  produces the edge
$$\tilde{\sigma}_0(f_{X_N}(i_c))=\tilde{\sigma}_0(i)=i+1.$$

A similar argument applies to $\sigma_1$, showing that the monodromies are indeed compatible. It now follows from covering space theory that $\beta=\tilde{\beta}\circ f_{X_N}$ , see \cite{DanishNotes} for example.

Let us now consider the general case of $X\in\mathcal{X}_{(N,k)}$ for a doubly even code $\mathcal{C}_k$ and let $A$ denote the Adinkra out of which $X$ is constructed. Since $\mathcal{C}_k\subseteq C_N$, the quotient of $A$ by $C_N$ is well-defined and we have
$$A/C_N\cong A_N/C_N.$$

This induces a well-defined map $f_X\colon X\to B_N$. The map $f_X$ sends all the white vertices to $0$, sends all  the black vertices to $\infty$, and identifies all edges of a given color $i$ with the edge of color $i$ in $\Sigma_N$.

We can argue exactly as we did earlier to show that the monodromy actions are compatible. The factorization of $\beta$ follows.
\end{proof}

It would be nice to explicitly see the action of $f_X$ on the faces of $X$. The monodromy $\sigma_\infty$ is what describes this action. However, under the projection induced by $f_X$, all of the disjoint $2$-cycles in $\sigma_\infty$ are sent to the identity element since all edges of the same color are identified via $f_X$. Despite the appearance of losing information about the faces, this is actually an important observation. We remarked earlier that $\beta$ has order $N$ ramification over $0$ and $1$ and order $2$ ramification over $\infty$, while $\tilde{\beta}$ has order $N$ ramification over $0$ and $1$ and is unramified elsewhere. That $\tilde{\beta}$ is unramified over $\infty$ is seen from the fact that $\sigma_\infty$ projects trivially. Furthermore, we see that all of the ramification of $X$ over the Belyi basis is split, so that all of the order $N$ ramification (over $0$ and  $1$) occurs in $\tilde{\beta}\colon B_N\to\mathbb{C}\mathbb{P}^1$ and all of the order $2$ ramification occurs in the map $f_X\colon X\to B_N$ at the centers of the $2$-colored faces. 

The data of the faces of $B_N$ is better represented by a certain element $\tilde{\pi}_\infty$ that lists the edges of the faces going clockwise. For $B_N$ with rainbow $(1,2,\dots,N)$, we can label the edges with the numbers $1$ through $N$ in the obvious manner. Then
$$\tilde{\pi}_\infty=\prod_{i=1}^N(i,i+1).$$

All of the faces of $B_N$ are $2$-gons, and  the transpositions in $\tilde{\pi}_\infty$ simply list the edges that make up each bi-gon. Recall from the paragraph before Proposition 1 that the element $\pi_\infty\in S_I$ describe the faces of $X$, and that each $4$-cycle in $\pi_\infty$ lists the edges of the faces as we move around clockwise. The map $f_X$ identifies edges of the same color in $\pi_\infty$. Therefore, forgetting about the last two entries and applying the natural map $S_I\to S_N$ induced by $i_c\mapsto i$, we obtain $\tilde{\pi}_\infty$ from $\pi_\infty$. Therefore, we can view the action of $f_X$ on the faces in two parts. First, each $2$-colored $4$-gon in $X$ is  mapped to a face with two sides by identifying the points opposite to each other (accounting for the order $2$ ramification); then all of the $2$-gons with the same $2$-color boundary are identified. 

Note that from this point of view, we see that the information of the rainbow is completely contained in $B_N$. The rainbow is what determines which $2$-colored loops of $A$ are filled in to create $X$, but the faces can be viewed as the pre-images of the faces of $B_N$. Fixing the rainbow for $B_N$ fixes which $2$-colored $2$-gons of $\Sigma_N$ are filled in, and this determines which faces in $A$ are filled in. Lastly, $B_N$ is the first possible place the rainbow can be seen, since it is where the edge in the Belyi base first splits into $N$ colored edges.

If we consider the Riemann surfaces $X_N$ and $X\in\mathcal{X}_{(N,k)}$ with rainbow $(1,2,\dots, N)$, we have shown the existence of maps $f_{X_N}\colon X_N\to B_N$ and $f_X\colon X\to B_N$ that fit into the following commutative diagram:

$$\xymatrix{X_N\ar_{\beta_{X_N}}[ddr]\ar^{f_{X_N}}[dr]&&X\ar_{f_X}[dl]\ar^{\beta_X}[ddl]\\ 
& B_N\ar^{\tilde{\beta}}[d]&\\
&\mathbb{C}\mathbb{P}^1&}$$

We will now show that  the map $f_{X_N}\colon X_N\to B_N$ factors through $f_X$. In order to accomplish this, we will argue similarly as was done for the factorization through $B_N$. We will work with the monodromy groups of these Riemann surfaces over $B_N$ and show that they are compatible. In order to do this, we will first need to describe the monodromy groups of each of the maps $f_X$. Note that since the maps $f_X$ are unramified over the vertices, we can describe the monodromies as permutations of the white vertices; we could not define the monodromies of the Belyi map in terms of the vertices because the Belyi map was ramified there.

\begin{theorem}
\label{Thm:monBN}
The monodromy group of $(X,f_X)$ is described by the elements
$$\rho_i=\prod_{c\in H_i} (c,c+c_i)$$
of $S_{C/\mathcal{C}_k}$, $1 \leq i \leq N$,
where $H_i$ is a set of orbit representatives for the action of $\gen{c_i}$ on  $C/\mathcal{C}_k$.

The monodromy group can be generated by $N-k-1$ elements. 

\end{theorem}
\begin{proof}
The analysis above shows that $f_X\colon X\to B_N$ is an order $2^{N-k-1}$ covering map with order $2$ ramification at the centers of the $2$-faces. Note that the ramification is over the roots of $-1$ in $B_N$. Therefore, we can describe the monodromy group of $f_X$ by giving the generators that are the monodromies over the centers of the $N$ faces of $B_N$. Assign the label $i_{f_X}$ to the center of the $i/(i+1)$ $2$-colored face and consider the loop based at the white vertex that travels along the edge of color $i$ and returns along the edge of color $i+1$. This loop lifts via $f_X$ to a path starting at a white vertex $c$ and ending at the white vertex $c+c_i$. The lift of the loop with starting point $c+c_i$ returns us to $c$, whence the result.

If we look only at the monodromy group of $f\colon X_N\to B_N$, then we find that
$$\prod_{i=1}^N\rho_i=1$$ since $c_1+\cdots+c_N=0$. Therefore, the monodromy group for $f$ can be generated by $N-1$ elements. This can also be seen from the fact that the path obtained by traveling along all of the $i/(i+1)$ loops in succession is null-homotopic in $B_N$ punctured at the centers of the $N$ faces.

In the general case, let $v_1,\dots, v_k$ be  generators of $\mathcal{C}_k$. Writing each of the $v_i$ as a linear combination of the elements $c_1,\dots, c_{N-1}$ in $C$, we obtain $k+1$ relations among the $\rho_i$, showing that the monodromy group can be generated by $N-k-1$ elements. 
\end{proof}

\begin{theorem}
\label{Thm:AdinkFact}
The pair $(X_N,f_{X_N})$ factors through $(X,f_X)$ for any $X\in\mathcal{X}_{(N,k)}$.
\end{theorem}
\begin{proof}

Let $p\colon X_N\to X$ be the natural projection induced by the map $A_N\to A=A_N/\mathcal{C}_k$ for a doubly even code $\mathcal{C}_k$. We will argue that the monodromy actions are compatible, from which the desired factorization follows. We need to show that the following diagram commutes, where $w$ is a white vertex and  $B_N^*$ denotes $B_N$ with the centers of the faces removed:

$$\xymatrix{\pi_1(B_N^*)\times f_{X_N}^{-1}(w)\ar[r]\ar[d]_{\mathrm{id}\times p}&f_{X_N}^{-1}(w)\ar[d]^{p}\\
\pi_1(B_N^*)\times f_X^{-1}(w)\ar[r]&f_X^{-1}(w)}$$

If we start with a white vertex $c$ and a monodromy generator $\rho_i$ and follow the diagram clockwise, we are left with
$$p(\rho_i(c))=p(c+c_i)=[c+c_i]$$
where $[\cdot]$ is being used to emphasize the fact that the right-hand side is now an equivalence class in $C/\mathcal{C}_k$. Similarly, if we follow the diagram counter-clockwise, we are left with
$$\cls{\rho_i}(p(c))=\cls{\rho_i}([c])=[c]+[c_i]=[c+c_i],$$
where $\cls{\rho}_i$ is the $i$-th monodromy generator for $X_{(N,k)}$.
\end{proof}

Note that in the above proof we see a new reason  we must quotient the $N$-cube  by a doubly even code to get an Adinkra, instead of an arbitrary code. The code has to be even to preserve the bipartite structure. However, if the code is not doubly even, then it is possible that the faces would  not map to faces of the correct form. For example, quotienting by the maximal even code produces bi-gon faces instead of $4$-gon faces. That $p$ maps the faces in the desired fashion above is seen from the fact that the two elements appearing in each $2$-cycle, $c$ and $c+c_i$, are distinct as elements of $C/\mathcal{C}_k$. In summary,   evenness  ensures the  preservation of the bipartite structure, and the double evenness ensures the proper action on the faces of the induced Riemann surfaces. 

\begin{corollary}
The Belyi pair $(X_N,\beta)$ factors through the Belyi pair $(X,\beta_X)$ for $X\in \mathcal{X}_{(N,k)}$. 
\end{corollary}

We now have the following diagram of morphisms:

$$\xymatrix{X_N\ar^p[d]\\ X\ar^{f_X}[d]\\ B_N\ar^{\tilde{\beta}}[d]\\ \mathbb{P}^1(\mathbb{C})}$$

The significance of Theorem \ref{Beachball} is that it now makes sense to study the Riemann surfaces in $\mathcal{X}_{(N,k)}$ as branched covers of $B_N$ rather than covers of $\mathbb{C}\mathbb{P}^1$. One reason  we would prefer this situation is that there is less ramification to worry about if we are working over $B_N$. In turn, Theorem \ref{Thm:AdinkFact} asserts that the curves $X$ are intermediate covers of $X_N\to B_N$. Therefore, the study of the curves lying in $\mathcal{X}_{(N,k)}$ is closely related to the study of the cover $X_N\to B_N$, with the elements of $\mathcal{X}_{(N,k)}$ corresponding to subgroups of the deck transformation group. 

We conclude this section by observing that one can build the monodromy of $\beta_X$ with the monodromies of $f_X$ and the rainbow. Indeed, $\sigma_0(X)$ can be built from the rainbow alone by indexing the white vertices by elements $c\in C/\mathcal{C}_k$ and labeling the edge of color $i$ at that vertex with $i_c$. 
On the other hand, $\sigma_\infty$ is completely determined by the $\rho_i(X)$. Recall that $\sigma_\infty$ is the product of the transpositions that list the two edges opposite to each other in  each face in $X$. The element $\rho_i$ is the product of the transpositions that  list the two white vertices that are opposite each other in each face. Therefore, if we simply associate to each $2$-cycle in $\rho_i$ the two edges of the same color incident to the two white vertices, we recover $\sigma_\infty$.
That is, if we make the association 
$$\rho_i=\prod_{c\in H_i}(c,c+c_i)\mapsto\prod_{c\in C_i} (i_c,i_{c+c_i}),$$ 
then taking the product over all the $\rho_i$ will gives us $\sigma_\infty$. This determines $\sigma_1$. 

\subsection{Invariance Under $R$-Symmetry}

In this section we will show that Adinkras  related by the permutation subgroup of $R$-symmetry give rise to equivalent Belyi curves. As discussed in Section \ref{sec:RevAd}, the permutation subgroup of the full $R$-symmetry group $O(N)$ is the largest subgroup whose action on an Adinkra is well defined. Therefore, abusing terminology, we will refer to the permutation subgroup as an $R$-symmetry. Let us now recall what it means for Belyi curves to be equivalent.

\begin{definition}
Two Belyi pairs $(X_1,\beta_1)$ and $(X_2,\beta_2)$ are equivalent if there is an isomorphism $X_1\to X_2$ that commutes with the Belyi maps, that is, they are equivalent as branched covers.
\end{definition}

Branched covers are equivalent when their monodromy groups are conjugate; see for example \cite{Girondo:Text}. We will use this fact  to show that the Riemann surfaces obtained from $R$-symmetric Adinkras are equivalent. 

In Section \ref{sec:GalCov} we described the monodromy generators as elements of $S_I$, where $I=\{i_c| i=1,\dots, N, c\in C/C_k\}$. This was particularly useful because the focus was on quotients and which vertices (expressed as elements of $\mathbb{F}_2^N$) were identified. For many computational applications it necessary to view the monodromy generators as elements of $S_{2^{N-k-1}N}$. The Belyi map has degree $2^{N-k-1}N$. It is unramified on the edges of the Adinkra. Furthermore, $R$-symmetry affects the edges of an Adinkra while leaving the vertices unchanged, making it useful to describe the monodromy generators as permutations of the $2^{N-k-1}N$ edges. To that end, we will first describe an isomorphism from $S_I$ to $S_{2^{N-k-1}N}$. 

Let us start with the $N$-cube. Consider the map $\gamma\colon \mathbb{F}_2^N\to\mathbb{N}$ defined by 
\begin{equation}
(a_1,\dots, a_N)\mapsto 1+a_2 2^0+a_3 2^1+\cdots a_N 2^{N-2}.
\end{equation}
If we view the elements of $\mathbb{F}_2^N$ as binary numbers, the map $\gamma$ is almost exactly the standard map of the binary numbers into the natural numbers. The differences are the removal of the dependence on $a_1$ and the translation by $1$. We have chosen to remove the dependence on $a_1$ so that we can number the white vertices (even elements of $\mathbb{F}_2^N$) and black vertices (odd elements of $\mathbb{F}_2^N$) such that two vertices have the same label if they are connected by an edge of color $1$. We have shifted by $1$, so  no vertex is labeled $0$. Following the above discussion, it is easy to see that the restriction of $\gamma$ to either the odd or the even elements of $\mathbb{F}_2^N$ is a bijection onto the set $\{1,\dots, 2^{N-1}\}$.
\begin{proposition}
\label{prop:label}
Suppose we give every white vertex $c\in C$ the label $\gamma(c)_w$ and every black vertex $d\in D$ the label $\gamma(d)_b$. Then $i_w$ is joined to $i_b$ by color $1$, and it is joined to $(i+2^{k-2})_b$ by color $k$ if $1\leq i\  (\bmod\  2^{k-1})\leq 2^{k-2}$ and to $(i-2^{k-2})_b$ otherwise. 
\end{proposition}
\begin{proof}
It is easy to see that if $c$ and $d$ differ by $e_1$, that is, if $c$ is joined to $d$ by color $1$, then $\gamma(c)=\gamma(d)$. Further, for $k\geq 2$ we find that
$$|\gamma(c)-\gamma(d)|=2^{k-2}$$
if $c$ and $d$ differ by $e_k$. 
Lastly, after a quick check we find that $\gamma(d)>\gamma(c)$ precisely when
$$1\leq \gamma(c)\  (\bmod\  2^{k-1})\leq 2^{k-2}.$$
\end{proof}

Suppose that we have labeled the vertices of $A_N$  as above. We can then label the edge of color $k$ incident to $i_w$ with the number $(i-1) N+k$. Such a labeling gives rise to the following monodromies.

\begin{corollary}
\label{cor:cubemon}
The edges of the $N$-cube $A_N$ can be labeled so that the monodromy group corresponding to the Belyi map is generated by 
$$\sigma_0=(1\;2\cdots N)(N+1\;N+2\cdots 2N)\cdots((m-1)N+1\;(m-1)N+2\cdots mN)$$
and
$$\sigma_1=(a_N^{(1)}\;a_{N-1}^{(1)}\cdots a_1^{(1)})\cdots(a_N^{(m)}\;a_{N-1}^{(m)}\cdots a_1^{(m)}),$$
where $m=2^{k-1}$, $a_1^{(i)}=(i-1)N+1$, and for $k\neq1$
\begin{equation}
\label{eq:ak}
a_k^{(i)} = \left\{
        \begin{array}{ll}
          (2^{k-2}+i-1)N+k   & \text{ if }1 \leq i \ (\mathrm{mod} \ 2^{k-1}) \leq 2^{k-2}, \\
           \strut (i-1-2^{k-2})N+k   & \text{ otherwise }
          \end{array}\right.
\end{equation}
\end{corollary}
\begin{proof}
The formula for $\sigma_0$ is immediately clear from the labeling: if we choose a white vertex $\gamma(c)_w,$ $c\in C$, and list the edges incident to it in the order of the rainbow ,we get $((\gamma(c)-1)N+1,(\gamma(c)-1)N+2,\ldots,\gamma(c)N)$. The result for $\sigma_1$ follows immediately from the previous formula for $\sigma_1$ as an element of $S_I$ by noting $i_c\mapsto (\gamma(c)-1)N+i$ in the new labeling. 
\end{proof}
For the sake of completeness, we also note the monodromy generators for the map $f: X_N\to B_N$.
\begin{corollary}
The monodromy group of $(X_N,f_{X_N})$ is generated by $\rho_i\in S_{2^{N-1}}$, $1\leq i\leq N-1$, where
\begin{align}
\rho_1 &=(1,2)(3,4)\cdots(2^{N-1}-1,2^{N-1}),\\
\rho_i &=\prod_{\substack{j=1\\1\leq j\bmod{2^{i-1}}\leq 2^{i-2}}}^{2^{N-1}}(j,j+2^{i-2}+2^{i-1}) \text{ for }i\neq 1.
\end{align}
\end{corollary}
\begin{proof}
This follows immediately from the formula for $\rho_i$ in terms of elements of $C$ by applying the map $\gamma$.
\end{proof}

The monodromy generators for a general Adinkra (quotient of the $N$-cube) follow immediately from the above formulae, if we use the smallest representatives for the equivalence classes in the quotienting procedure. In this procedure, the formulae for $\sigma_i$ will be the same as for the cube but with $m=2^{N-k-1}$ instead of $2^{N-1}$. An example will help illustrate this.

\begin{example}[$N=4$]

Consider $A_4$ with rainbow $(1,2,3,4)$ and corresponding Belyi pair $(X_4,\beta)$. By Propostion \ref{prop:genus},  $X_4$ has genus $1$ and is therefore an elliptic curve (its complex structure is pulled back from $\bP^1(\bC)$ by $\beta$); see Figure \ref{Fig:N4embeds}. 
\begin{figure}
\begin{tikzpicture}[baseline, scale=2]
\draw[green] (-0.9375,-1) -- (-0.5625,-1);
\draw[orange] (-0.4375,-1) -- (-0.0625,-1);
\draw[green] (0.0625,-1) -- (0.4375,-1);
\draw[orange] (0.5625,-1) -- (0.9375,-1);
\draw[green] (-0.9375,-0.5) -- (-0.5625,-0.5);
\draw[orange] (-0.4375,-0.5) -- (-0.0625,-0.5);
\draw[green] (0.0625,-0.5) -- (0.4375,-0.5);
\draw[orange] (0.5625,-0.5) -- (0.9375,-0.5);
\draw[green] (-0.9375,0) -- (-0.5625,0);
\draw[orange] (-0.4375,0) -- (-0.0625,0);
\draw[green] (0.0625,0) -- (0.4375,0);
\draw[orange] (0.5625,0) -- (0.9375,0);
\draw[green] (-0.9375,0.5) -- (-0.5625,0.5);
\draw[orange] (-0.4375,0.5) -- (-0.0625,0.5);
\draw[green] (0.0625,0.5) -- (0.4375,0.5);
\draw[orange] (0.5625,0.5) -- (0.9375,0.5);
\draw[green,dashed] (-0.9375,1) -- (-0.5625,1);
\draw[orange,dashed] (-0.4375,1) -- (-0.0625,1);
\draw[green,dashed] (0.0625,1) -- (0.4375,1);
\draw[orange,dashed] (0.5625,1) -- (0.9375,1);
\draw[blue] (-1,-0.9375) -- (-1,-0.5625);
\draw[purple] (-1,-0.4375) -- (-1,-0.0625);
\draw[blue] (-1,0.0625) -- (-1,0.4375);
\draw[purple] (-1,0.5625) -- (-1,0.9375);
\draw[blue] (-0.5,-0.9375) -- (-0.5,-0.5625);
\draw[purple] (-0.5,-0.4375) -- (-0.5,-0.0625);
\draw[blue] (-0.5,0.0625) -- (-0.5,0.4375);
\draw[purple] (-0.5,0.5625) -- (-0.5,0.9375);
\draw[blue] (0,-0.9375) -- (0,-0.5625);
\draw[purple] (0,-0.4375) -- (0,-0.0625);
\draw[blue] (0,0.0625) -- (0,0.4375);
\draw[purple] (0,0.5625) -- (0,0.9375);
\draw[blue] (0.5,-0.9375) -- (0.5,-0.5625);
\draw[purple] (0.5,-0.4375) -- (0.5,-0.0625);
\draw[blue] (0.5,0.0625) -- (0.5,0.4375);
\draw[purple] (0.5,0.5625) -- (0.5,0.9375);
\draw[blue,dashed] (1,-0.9375) -- (1,-0.5625);
\draw[purple,dashed] (1,-0.4375) -- (1,-0.0625);
\draw[blue,dashed] (1,0.0625) -- (1,0.4375);
\draw[purple,dashed] (1,0.5625) -- (1,0.9375);
\draw (0,0) circle[radius=0.0625];
\draw (-1,-1) circle[radius=0.0625];
\draw (0.5,0.5) circle[radius=0.0625];
\draw[dashed] (1,1) circle[radius=0.0625];
\draw (0,-1) circle[radius=0.0625];
\draw (-1,0) circle[radius=0.0625];
\draw (-0.5,0.5) circle[radius=0.0625];
\draw[dashed] (-1,1) circle[radius=0.0625];
\draw[dashed] (0,1) circle[radius=0.0625];
\draw[dashed] (1,0) circle[radius=0.0625];
\draw[dashed] (1,-1) circle[radius=0.0625];
\draw (-0.5,-0.5) circle[radius=0.0625];
\draw (0.5,-0.5) circle[radius=0.0625];
\draw[fill] (-0.5,-1) circle[radius=0.0625];
\draw[fill] (0.5,-1) circle[radius=0.0625];
\draw[fill] (-1,-0.5) circle[radius=0.0625];
\draw[fill] (-0.5,0) circle[radius=0.0625];
\draw[fill] (0,-0.5) circle[radius=0.0625];
\draw[fill] (0.5,0) circle[radius=0.0625];
\draw[fill] (-1,0.5) circle[radius=0.0625];
\draw[fill] (0,0.5) circle[radius=0.0625];
\draw[dashed,fill=gray] (0.5,1) circle[radius=0.0625];
\draw[dashed,fill=gray] (-0.5,1) circle[radius=0.0625];
\draw[dashed,fill=gray] (1,0.5) circle[radius=0.0625];
\draw[dashed,fill=gray] (1,-0.5) circle[radius=0.0625];
\end{tikzpicture}
\caption{The $(4,0)$ Adinkra with rainbow (green,blue,orange,purple) embedded in a torus. The top row is identified with the bottom row and the right column is identified with the left column. The $(4,1)$ Adinkra embedding is obtained by taking only the left half of the $(4,0)$ embedding.}
\label{Fig:N4embeds}
\end{figure}
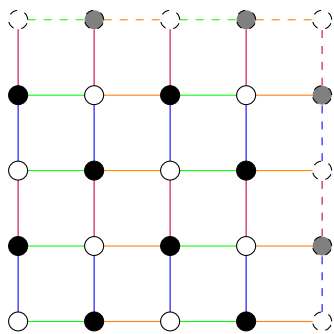
The monodromy group of $\beta$ is generated by
\begin{align}
\sigma_0 &=(1,2,3,4)(5,6,7,8)(9,10,11,12)(13,14,15,16)\\
&\cdot(17,18,19,20)(21,22,23,24)(25,26,27,28)(29,30,31,32)\nonumber
\end{align}
and
\begin{align}
\label{eq:4cubeSigma1}
\sigma_1 &=(20,11,6,1)(24,15,2,5)(28,3,14,9)(32,7,10,13)\\
&\cdot(4,27,22,17)(8,31,18,21)(12,19,30,25)(16,23,26,29).\nonumber
\end{align}
The faces are described by
\begin{align}
\label{eq:4cubePiInf}
\pi_\infty &=(1,2,5,6)(2,3,14,15)(3,4,27,28)(4,1,20,17)\\
&\cdot (6,7,10,11)(7,8,31,32)(8,5,24,21)(9,10,13,14)\nonumber\\
&\cdot (11,12,19,20)(12,9,28,25)(15,16,23,24)(16,13,32,29)\nonumber\\
&\cdot (17,18,21,22)(18,19,30,31)(22,23,26,27)(25,26,29,30).\nonumber
\end{align}

Now, assign the label $i$ to the edge of color $i$ incident to the white vertex in $B_4$. For $\tilde{\beta}:B_4\to \bP^1(\bC)$,
$$\tilde{\pi}_\infty=(1,2)(2,3)(3,4)(4,1).$$
The map $f_{X_4}:X_4\to B_4$ projects $\sigma_0$ onto $(1,2,3,4)$ and $\sigma_1$ onto $(4,3,2,1)$. More interesting is how we can see the action on the faces by looking at the action of $f$ on $\pi_\infty$. All of the $1/2$-colored faces, $(1,2,5,6)$, $(9,10,13,14)$, $(25,26,29,30)$, and $(17,18,21,22)$ in $X_4$ are projected onto the $1/2$ colored face $(1,2)$ in $B_4$ via $f_{X_4}$. Similarly the four $i/(i+1)$ faces in $X_4$ map onto the single $i/(i+1)$ face in $B_N$, what can be realized as the map that sends each number $j$ appearing in each $4$-cycle to $j\pmod{4}$. Here we represent $0\pmod{4}$ by $4$.

Finally, we consider the generators of the monodromy of $f$:
\begin{align}
\label{eq:rho14cube}
\rho_1 &=(1,2)(3,4)(5,6)(7,8),\\
\rho_2 &=(1,4)(2,3)(5,8)(6,7),\label{eq:rho24cube}\\
\rho_3 &=(1,7)(2,8)(3,5)(4,6),\label{eq:rho34cube}\\
\rho_4 &=(1,5)(2,6)(3,7)(4,8).
\label{eq:rho44cube}
\end{align}
The numbers appearing in the $\rho_i$ represent the eight white vertices in $X_4$. Each $2$-cycle represents a single face that the two vertices listed are incident to. The element $\rho_i$ represents the four $i/(i+1)$ faces that are equated by $f$. Note that $\rho_4$, which represents the $4/1$ faces, is given by $\rho_1\rho_2\rho_3$.

Now let us consider the Belyi curve $X$ obtained by quotienting $A_4$ by the unique doubly even code generated by $(1,1,1,1)$. We denote this projection by $p_A:A_4\to A=A_4/\langle(1,1,1,1)\rangle$. The code generated by $(1,1,1,1)$ has two elements, so it divides the white vertices of $A_4$ into four cosets, each containing the two elements $\gamma(c)$ and $\gamma(c+(1,1,1,1))$. Explicitly, the cosets in $A_{(4,1)}$ are 
\begin{equation}
\label{eq:cosets41}
\{1,8\},\{2,7\},\{3,6\},\{4,5\}. 
\end{equation}
We see immediately that the quotienting procedure identifies $i_w$ with $(9-i)_w$ in $A_4$, and similarly for the black vertices. Each face of $X$ is the image under $p_A$ of two faces in $X_4$. We can see this from the effect of $p_A$ on the monodromy generators over $B_N$ that describe the faces.

We obtain $\rho_i(X)$ from $\rho_i(X_4)$ by identifying the $2$-cycles according to the action of $p_A$. For example, consider $\rho_1(X_4)=(1,2)(3,4)(5,6)(7,8)$, as in equation (\ref{eq:rho14cube}). From this we see that $(1,2)$ and $(7,8)$ should be equated and that $(3,4)$ and $(5,6)$ should be equated. If we choose the smaller number in each coset as the representative, we find
\begin{equation}
\rho_1(X)=(1,2)(3,4).
\end{equation}
Now consider equation (\ref{eq:rho24cube}), $\rho_2(X_4)=(1,4)(2,3)(5,8)(6,7)$. From equation (\ref{eq:cosets41}) we see that $(1,4)$ and $(5,8)$ should be equated and that $(2,3)$ and $(6,7)$ should be equated. Furthermore, since we have already chosen the representatives of the cosets we find
\begin{equation}
\rho_2(X)=(1,4)(2,3).
\end{equation}
Similarly we find
\begin{align}
\rho_3(X)=(1,2)(3,4),\\
\rho_4(X)=(1,4)(2,3).
\end{align}
The monodromy group of $X_{(4,1)}$ over $B_N$ has only two generators, since quotienting out by a $1$-dimensional code defines a relation between the generators of the monodromy group for $X_4$. The element $\rho_1(X_{(4,1)})$ interchanges white vertices that are connected by $(1,1,0,0)$, while $\rho_3(X)$ interchanges white vertices connected by $(0,0,1,1)$. Therefore, since $(1,1,0,0)+(0,0,1,1)=(1,1,1,1)$, the generator of the code, we see that $\rho_1(X)=(\rho_3(X))^{-1}$.

By Proposition \ref{prop:genus}, $X_{(4,1)}$ has genus $1$ and therefore is also a torus. If we give $X_4$ as depicted in Figure \ref{Fig:N4embeds} standard $(x,y)$ coordinates, then $p:X_4\to X_{(4,1)}$ is given by $p(x,y)=(2x,x+y)$. That is, it is the map from the torus to the torus that wraps around the diagonal twice, as can be seen from the $2$-to-$1$ effect on the faces by $p_A$. The left half of Figure \ref{Fig:N4embeds} can be taken as a fundamental domain for $X$.

The choice of coset representatives for the vertices determines representatives for the equivalence classes of the edges and similarly allows us to obtain $\sigma_i(X)$. Each equivalance class of edges contains two elements,  $(i-1)\cdot 4+j$ and $(8-i)\cdot 4+j$. Following our choice of representatives for the vertices, we choose the smaller number as the representative. The monodromy representatives for $X$ coming from those for $X_4$ are then as follows:
\begin{equation}
\sigma_0{(X)}=(1,2,3,4)(5,6,7,8)(9,10,11,12)(13,14,15,16)
\end{equation}
and
\begin{equation}
\sigma_1{(X)}=(1,16,11,6)(2,5,12,15)(3,14,9,8)(4,7,10,13).
\end{equation}
As an example, the $4$-cycle $(1,16,11,6)$ appearing in $\sigma_1(X)$ is the chosen representative of the coset consisting of the two elements $(1,20,11,6)$ and $(29,16,23,26)$; see equation (\ref{eq:4cubeSigma1}).
\end{example}

Now let us turn our attention to effect of the $R$-symmetry group on the Belyi curves. Recall that $R$-symmetry permutes the action of the $Q_i$. As a concrete example, consider $X_N$ with rainbow $(1,2,\ldots,N)$ and the $R$-symmetry that interchanges the action of $Q_1$ and $Q_3$. 
The conditions $Q_1\phi_0=\psi_1$ and $Q_3\phi_0=\psi_3$ is represented in $A_N\ins X_N$ by joining the white vertex representing $\phi_0$ to to the black vertex representing $\psi_1$ by color $1$, and to the black vertex representing $\psi_3$ by color $3$. 
 In our labeling scheme, we can identify $\phi_0$ with the vertex $1_w$ and $\psi_i$ with the vertices $i_b$.

In the  normal interpretation of $R$-symmetry, we would represent the $R$-symmetry that interchanges the actions of $Q_1$ and $Q_3$ by swapping the colors of the edges of colors $1$ and $3$. This makes sense; indeed,  after the $R$-symmetry, $Q_1\phi_0=\psi_3$ and $Q_3\phi_0=\psi_1$. The $Q_i$ remain represented by color $i$ and therefore $1_w$ should now be connected to $3_b$ by color $1$ and to $1_b$ by color $3$. This yields a new chromotopology; call this chromotopology $\tilde{A}_N$. Note that the rainbow is unchanged by $R$-symmetry. Under the $R$-symmetry, we changed only the actions of two of the $Q_i$, not their ordering.

In fact we should really view the rainbow as labeling which $Q_i$ each color represents. Since the rainbow determines which $2$-colored loops we attach $2$-cells to in order to create $\tilde{X}_N$, we are still attaching $2$-cells based on the same adjacent colors. The $2$-cells that are attached to loops containing only colors that are not adjacent to $1$ or $3$ will remain unchanged. Let us look at the $1/N$ $2$-colored faces as an example. By construction, the $1/N$ $2$-colored loops in $\tilde{A}_N$ are equivalent to $3/N$ $2$-colored loops in $A_N$. Therefore, the Riemann surface $\tilde{X}_N$ is isomorphic to the surface obtained from the original $A_N$, leaving the chromotopology $A_N$ unchanged, but using the rainbow $(3,2,1,4,\dots,N)$. We interpret this as interchanging which $Q_i$ the different colors represent. We are able to do this because we have the extra information of the order of the $Q_i$ contained in the rainbow. In this way, we may choose to view the action of $R$-symmetry as one that leaves the chromotopology invariant, but permutes the rainbow. 

To explicitly see that the equivalent descriptions of $R$-symmetric surfaces described above are compatible with the Belyi map, consider $\tilde{A}_N$ with the same labeling as $A_N$ so that the edges labeled $kN+1$ in $A_N$ are now color $3$ in $\tilde{A}_N$ (as opposed to color $1$ in $A_N$) and the edges labeled $kN+3$ in $A_N$ are now color $1$ in $\tilde{A}_N$. Therefore the monodromy at each white vertex $i_w$ in $\tilde{A}_N$, which is given by the edges incident to $i_w$ in the order of the rainbow $(1,2,3,\ldots,N)$, is $((i-1)N+3,(i-1)N+2,(i-1)N+1,(i-1)N+4,\ldots,iN)$. This is the same as the monodromy at each white vertex for $A_N$ with rainbow $(3,2,1,4,\dots,N)$. A similar argument shows that the monodromies at the black vertices coincide, showing that the Belyi pairs are equivalent. This shows that we can view the action of the permutation subgroup of the $R$-symmetry group on an Adinkra chromotopology with a rainbow as a permutation of the rainbow.

Let $C$ be the subset of even elements of $\mathbb{F}_2^N$, and consider the subset of $C$ consisting of elements that have zero as their second component. These are the white vertices that  map to an odd number under $\gamma$. Let $$S_e=\{\gamma(a_1,a_2,\ldots,a_N)\in \gamma(C)|a_1=a_2=0\}$$ and $$S_o=\{\gamma(a_1,a_2,\ldots,a_N)\in \gamma(C)|a_1=1,a_2=0\}.$$ We now prove that $R$-symmetry leaves the Belyi curves invariant.

\begin{theorem}
\label{Thm:cubeRsym}
The Belyi pair $(X_N,\beta)$ associated to the $N$-cube $A_N$ is invariant under the action of the permutation subgroup of the $R$-symmetry group.  
\end{theorem}
\begin{proof}
Give $A_N$ the edge labeling from Corollary \ref{cor:cubemon} with rainbow $(1,2,\ldots,N)$. The Adinkra $\tilde{A}_N$ has the the same underlying graph; all that is different is the rainbow, i.e., the cyclic order of the $N$ colors.  The difference between $A_N$ and $\tilde{A}_N$ is  encoded in a permutation of the order of the $N$ colors, i.e., an element of $S_N$. This means that for any two $N$-cubes $A_N$ and $\tilde{A}_N$ with rainbows $r\in S_N$ and $\tilde{r}\in S_N$ respectively, there exists $g\in S_N$ such that $\tilde{r}=grg^{-1}$. Since $(1,2)$ and $(1,2,\ldots,N)$ generate all of $S_N$ it is enough to show that the Belyi pair for $A_N$ with rainbow $r=(1,2,\ldots,N)$ is equivalent to the Belyi pair for $\tilde{A}_N$ with rainbow $\tilde{r}=grg^{-1}$ for just $g=(1,2)$ and $(1,2,\ldots,N)$.

Let us first consider the  case  $g=(1,2,\ldots,N)$. Conjugation by this $g$ in $S_N$ leaves the rainbow $r$ invariant, and therefore leaves $\sigma_0$ and $\sigma_1$ invariant as elements of $S_d$. The Belyi pairs are therefore trivially equivalent.

Now let us consider the case where $\tilde{A}_N$ is obtained from $A_N$ by interchanging the order of colors $1$ and $2$ in the rainbow. Since we have labeled the edges $A_N$ as in Corollary \ref{cor:cubemon}, we have $\sigma_0=\sigma_0^{(1)}\sigma_0^{(2)}\cdots\sigma_0^{(2^{N-1})}$, where $\sigma_0^{(i)}=((i-1)N+1\;(i-1)N+2\cdots iN)$, and $\sigma_1=\sigma_1^{(1)}\sigma_1^{(2)}\cdots\sigma_1^{(2^{N-1})}$, where $\sigma_1^{(i)}=(a_N^{(i)}\;a_{N-1}^{(i)}\cdots a_1^{(i)})$ with $a_k^{(i)}$ as in Corollary \ref{cor:cubemon}. We can give $\tilde{A}_N$ the same labeling as $A_N$ since they have the same underlying graphs, but the rainbow is now $\tilde{r}=(2,1,3,4,\ldots,N)$. Therefore the difference between $\sigma_q$ and $\tilde{\sigma}_q$ will be in the order of the colors. Changing the rainbow will leave unchanged the numbers appearing in each disjoint $N$-cycle $\sigma_0^{(i)}=((i-1)N+1,\;(i-1)N+2,\cdots iN)$,  but the order of each one mod $N$ will change from $(1,\;2,\cdots N-1,\;N)$ in the same way.
Therefore,
\begin{equation}
\tilde{\sigma}_0^{(i)}=((i-1)N+2,\;\;(i-1)N+1,\;\;(i-1)N+3,\cdots iN)
\end{equation}
and
\begin{equation}
\tilde{\sigma}_1^{(i)}=(a_N^{(i)},\;a_{N-1}^{(i)},\cdots a_3^{(i)},\;a_1^{(i)},\; a_2^{(i)}).
\end{equation}

Now define
\begin{equation}
\alpha=\prod_{\substack{i=1\\
\text{odd}}}^{2^{N-1}}(a^{(i)}_1,a^{(i+1)}_2)(a^{(i+1)}_1,a_2^{(i)}),
\end{equation}
\begin{equation}
\beta_o=\prod_{i\in S_o}(a_1^{(i)},a^{(i+1)}_1)(a_2^{(i)},a^{(i+1)}_2),
\end{equation}
and
\begin{equation}
\beta_e=\prod_{i\in S_e}\prod_{k=3}^N(a^{(i)}_k,a^{(i+1)}_k).
\end{equation}
Finally, let
\begin{equation}
\delta=\alpha\beta_o\beta_e.
\end{equation}
It is immediately clear that
$$\widetilde{\sigma}_1=\delta^{-1}\sigma_1\delta.$$
It is a little harder to see this for $\sigma_0$ since it is not written in terms of the$a_k^{(i)}$. However, using the formula for $a_k^{(i)}$ we see that
$$\alpha=\prod_{m=1}^{2^{N-1}}((m-1)N+1,(m-1)N+2),$$
while
$$\beta=\beta_o\beta_e=\prod_{m\in S_o}\prod_{k=1}^N((m-1)N+k,mN+k).$$
It follows  that
$$\widetilde{\sigma}_0=\delta^{-1}\sigma_0\delta,$$
completing the proof.
\end{proof}

Now let us consider $R$-symmetry for a general Adinkra\footnote{Again, when we refer to an $R$-symmetry we really mean the permutation subgroup of the full $R$-symmetry group $O(N)$.}. As noted in our discussion of $R$-symmetry for the $N$-cube, we can view an $R$-symmetry as leaving the chromotopology alone and changing the rainbow. There was nothing specific to the $N$-cube in that argument, so we see the same is true for general Adinkras. Recall that from this viewpoint we consider the extra data of the rainbow as encoding which supersymmetry generator each color represents rather than considering each supersymmetry generator fixed to the same color, as is usually the case. Since the rainbow is entirely determined by $B_N$, we can view the relationship between $R$-symmetric Adinkras as the pullback of the relationship between the $R$-symmetric $B_N$'s, where we know allow the possibility of different rainbows for the curve $B_N$.

\begin{proposition}
\label{prop:BNRsymm}  
If $\Sigma_N$ is related to $\tilde{\Sigma}_N$, by an $R$-symmetry, then the corresponding Belyi pairs $(B_N,\tilde{\beta})$ and $(\tilde{B}_N,\tilde{\beta})$ are equivalent.
\end{proposition}
\begin{proof}
We can consider the case where $B_N$ has the rainbow $(1,2,\dots,N)$. We just need to show that changing the rainbow by $(1,2)$ and $(1,2,\ldots,N)$ gives equivalent Belyi curves, since $(1,2)$ and $(1,2,\ldots,N)$ generate all of $S_N$. As with the case of the $N$-cube, changing the rainbow by $(1,2,\ldots,N)$ is trivial since the rainbow is invariant under conjugation by $(1,2,\ldots,N)$. The monodromy of $(B_N,\tilde{\beta})$ is 
\begin{align*}
\sigma_0=& (1,2,\ldots,N)
\end{align*}

and
\begin{align*} 
\sigma_1=&(N,N-1,\ldots,1).
\end{align*}

The Riemann surface $\tilde{B}_N$ has rainbow $(1,2)(1,2,\ldots,N)(1,2)=(2,1,3,\ldots,N)$. Therefore the monodromy of $(\tilde{B}_N,\tilde{\beta})$ is generated by
\begin{align*}
\tilde{\sigma}_0=& (2,1,3,\ldots,N)
\end{align*}
and
\begin{align*}
\tilde{\sigma}_1=&(N,N-1,\ldots,3,1,2).
\end{align*}
We see that $\tilde{\sigma}_q=(1,2)\sigma_q(1,2)$ for $q=0,1$. Therefore by the Grothendieck correspondence, $(B_N,\tilde{\beta})$ and $(\tilde{B}_N,\tilde{\beta})$ are equivalent Belyi pairs.
\end{proof}

Let us now suppose that $X\in\mathcal{X}_{(N,k)}$, corresponding to an Adinkra with rainbow $(1,2,\dots, N)$, and let $\tilde{X}$ be the Riemann surface associated to an $R$-symmetric Adinkra $\tilde{A}$.
By Theorems  \ref{Thm:AdinkFact} and \ref{Thm:cubeRsym} and Proposition \ref{prop:BNRsymm}, we have a commutative diagram
\begin{equation}
\xymatrix{  X_N \ar@{->}[d]_{p}\ar@{->}[r]^{\cong}& \tilde{X}_N \ar@{->}[d]^{\tilde{p}}\\ 
X\ar@{->}[d]_{g}& \tilde{X}\ar@{->}[d]^{\tilde{g}}\\
B_N\ar@{->}[r]^{\cong} & \tilde{B}_N  }
\end{equation}

\begin{proposition}
\label{prop:Rsymm}
If $A\in\mathcal{A}_{(N,k)}$ is an Adinkra and $\tilde{A}$ is related to $A$ by an $R$-symmetry, then the corresponding Belyi pairs $(X,\beta_X)$ and $(\tilde{X},\beta_{\tilde{X}})$ are equivalent.
\end{proposition}
\begin{proof}
As usual it is enough to consider the $R$-symmetry that interchanges the actions of $Q_1$ and $Q_2$. By assumption, $X$ and $\tilde{X}$ have the same chromotopology; they differ only in  their rainbows. Therefore we can label their vertices so that $\rho_i=\tilde{\rho}_i$. For example, if the white vertex $i$ is connected to vertex $j$ by an edge of color $2$ followed by an edge of color $3$ and is connected to vertex $l$ by an edge of color $1$ followed by an edge of color $3$, then the labels of $j$ and $l$ should be interchanged in $\tilde{X}$. Since the monodromies of $g$ and $\tilde{g}$ are the same and $(B_N,\tilde{\beta})$ and $(\tilde{B}_N,\tilde{\beta})$ are equivalent, the result follows.

In particular, there exists a pullback $g$ of $(1,2)$, such that $\tilde{\sigma}_q=g^{-1}\sigma_q g$. As an example, note that $\gamma$ in the proof of Theorem \ref{Thm:cubeRsym} is a pullback of $(1,2)$. We can see this by observing that $(1,2)=(1,2)(3)\cdots(N)$. For example, $(3)$ can lift to $(mN+3,(m+1)N+3)$, as is the case here.
\end{proof}


\subsection{The Exterior Tensor Product}
Now that we have shown $R$-symmetric Adinkra chromotoplogies give equivalent Belyi pairs, we can extende the exterior tensor product on Adinkras to an operation on the corresponding Riemann surfaces.

The tensor product of two Adinkras is defined in \cite{Hubsch:2011tx}. If $A_i\in\mathcal{A}_{(N_i,k_i)}$ is an Adinkra chromotopology obtained by quotienting an $N_i$-cube by a $k_i$-dimensional doubly even code, then
$$A_1\otimes A_2$$ is the Adinkra chromotopology obtained by quotienting an $(N_1+N_2)$-cube by a $(k_1+k_2)$-dimensional code.
More precisely, we consider the direct sum  $\mathbb{F}_2^{N_1}\oplus\mathbb{F}_2^{N_2}$, which is canonically isomorphic as a group to $\mathbb{F}_2^{N_1+N_2}$. Under this isomorphism, the image of the $(k_1+k_2)$-dimensional code $C_{k_1}\oplus C_{k_2}\ins\mathbb{F}_2^{N_1}\oplus \mathbb{F}_2^{N_2}$ is a doubly even code, and we define the tensor product Adinkra to be the quotient of $\mathbb{F}_2^{N_1+N_2}$ by this code. Although it may seem more appropriate to refer to this construction as a direct sum, the accepted term is tensor product, which stems from the graph-theoretic roots of the construction.

For convenience, let $N=N_1+N_2$ and $k=k_1+k_2$. The white vertices of $A_1\otimes A_2$ correspond to the orbits of the even elements of $\mathbb{F}_2^N$ under the action of $C_k$, and the black vertices correspond to the orbits of the odd elements. Therefore, there are two types of white vertex in $A_1\otimes A_2$: we can have a white vertex of the form $w_1\oplus  w_2$ where $w_i\in A_i$ are white vertices, or we can have a white vertex of the form $b_1\oplus  b_2$ where $b_i\in A_i$ are black vertices. The black vertices of $A_1\otimes A_2$  arise as the  sum of two vertices of different colors in the original Adinkras. 

Now let us describe the rainbow and edges of $A_1\otimes A_2$. Suppose that $A_i$ has rainbow $(1,2,\dots, N_i)$. We can give $A_1\otimes A_2$ the rainbow $(1,2,\dots,N_1,N_1+1,\dots, N_1+N_2)$. The edges are described as follows. Let $v\oplus w$ be a vertex; the edge of color $i$ incident to this vertex for $1\leq i\leq N_1$ is the edge that connects $v\oplus w$ to $v'\oplus w$, where $v'$ is the vertex in $A_1$ that is joined to $v$ by color $i$. Similarly, if $N_1+1\leq i\leq N_1+N_2$, the edge of color $i$ is the edge that joins $v\oplus w$ to $v\oplus w'$, where $w'$ is the vertex in $A_2$ that is joined to $w$ by color $i-N_1$. Observe that if we delete all of the edges of $A_1\otimes A_2$ of colors $N_1+1\leq i\leq N_1+N_2$, we are left with $2^{N_2-k_2}$ disjoint copies of $A_1$, one copy for each vertex of $A_2$. Similarly, if we delete all of the edges of colors $1\leq i\leq N_1$, then we obtain $2^{N_1-k_1}$ disjoint copies of $A_2$. 

Lastly, we remark that the dashing of the Adinkra $A_1\otimes A_2$ is obtained from the dashings of the two Adinkras as follows: the dashing at a vertex corresponding to two vertices of the same color is left alone, while the dashing at a vertex corresponding to two vertices of different colors is reversed. 

\begin{proposition}
The exterior tensor product on Adinkras extends to a well-defined operation on the associated Belyi curves. 
\end{proposition}

\begin{proof}
Since a Belyi curve is equivalent to a ribbon graph, it suffices to show that the tensor product extends to a well-defined operation on the ribbon structure, or rainbow, of the Adinkra chromotopology. The chromotopology of $A_1\otimes A_2$ is defined by the tensor product. The rainbow is a cyclic ordering of the $N=N_1+N_2$ colors, that ordering being determined by the rainbows of the $A_i$. If $r_1$ and $r_2$ are two rainbows obtained by combining the rainbows for the $A_i$ in two different ways, then they are related by an $R$-symmetry. Therefore, they are equivalent Belyi curves by Proposition \ref{prop:Rsymm}.

It follows that the tensor product on Adinkras extends to the associated Belyi pairs by assigning any cyclic ordering of the $N$ colors to the product as a rainbow. We will denote by $X_1\otimes X_2$ the Riemann surface constructed out of the Adinkra $A_1\otimes A_2$. 
\end{proof}

For concreteness, we will use the choice of rainbow for $A_1\otimes A_2$ that was given in the above discussion, namely $(1,2,\dots, N_1+N_2)$. 

Let us now examine what this operation looks like on the corresponding Riemann surfaces $X_1$ and $X_2$. It is useful to first look at the operation on the monodromies of $X_i$ over $B_{N_i}$. Before stating our theorem, we fix some notation. Let $C_i,D_i\ins\mathbb{F}_2^{N_i}$ denote the set of even and odd elements respectively, and let $X_i$ be the surface obtained by quotienting the $N_i$-cube by a $k_i$-dimensional doubly even code $\mathcal{C}_i$, so that the elements of $C_i/\mathcal{C}_{k_i}$ correspond to the white vertices of $A_i$ and the elements of $D_i/\mathcal{C}_{k_i}$ correspond to the black vertices. As usual, let $c_j$ be the element that has zero in every position except for the $j$ and $j+1$ positions. Then the monodromies of $X_{(N_i,k_i)}$ are given by
$$\rho_{j,i}=\prod_{c\in H_j}(c,c+c_j),$$ 
where $H_j$ is a set of orbit representatives for the action of $\gen{c_j}$ on $C_i/\mathcal{C}_{k_i}$. Note that we could have defined these monodromies in terms of the black vertices, in which case they would take the form
$$\rho_{j,i}'=\prod_{d\in K_j}(d,d+c_j),$$ 
where $K_j$ is a set of orbit representatives for the action of $\gen{c_j}$ on $D_i/\mathcal{C}_{k_i}$. We will use $\rho_{j,i}$ to denote the $j$-th monodromy generator for $X_i$ written in terms of white vertices, and we will use $\rho_{j,i}'$ to denote the $j$-th monodromy generator for $X_i$ written in terms of black vertices.

Finally, it will be convenient for us to adopt the following notational convention. If $v$ is a vertex of $A_2$ and $\tau$ is a permutation of the vertices of $A_1$ written as a cycle, then we define  $\tau\oplus v$ to be the cycle with each entry equal to the corresponding entry of $\tau$ tensored by $v$ on the right; we define $v\oplus \tau$ similarly. For example, if $(c,d)$ is a $2$-cycle with $c,d$ being white vertices in $X_1$, and $w$ is a white vertex in $X_2$, then $(c,d)\oplus w=(c\oplus w,d\oplus w)$, the right-hand side now being a permutation of the white vertices of $X_1\otimes X_2$. 

\begin{theorem}
\label{Thm:rhoTen}
The monodromy group of $X_1\otimes X_2$ over $B_N$ is generated by elements $\rho_j$, $1 \leq j \leq N_1 + N_2$, given as follows.

If $1\leq j\leq N_1-1$, then 

$$\rho_j=\prod_{v\in V_2}\rho_j^{(v)}$$
where $V_2$ is the set of all vertices in $A_2$ and

$$\rho_j^{(v)}=\left\{\begin{array}{ll}
\rho_{j,1}\oplus v&\textrm{ if $v$ is a white vertex}\\
\rho_{j,1}'\oplus v&\textrm{ if $v$ is a black vertex}.
\end{array}
\right.$$

If $j=N_1$, then 
$$\rho_{N_1}=\prod_{w_1\in W_1,w_2\in W_2}\left(w_1\oplus w_2,(w_1+e_{N_1})\oplus(w_2+e_1)\right)$$
where $W_i$ is the set of white vertices in $A_i$.

If $N_1+1\leq j\leq N_1+N_2-1$, then
$$\rho_j=\prod_{v\in V_1}\rho_j^{(v)}$$ where $V_1$ is the set of all vertices in $A_1$ and

$$\rho_j^{(v)}=\left\{\begin{array}{ll}
v\oplus \rho_{j,2}&\textrm{ if $v$ is a white vertex}\\
v\oplus \rho_{j,2}'&\textrm{ if $v$ is a black vertex}.
\end{array}
\right.$$

Finally,
$$\rho_N=\prod_{\substack{w_1 \in W_1 \\ w_2 \in W_2}}(w_1\oplus w_2,(w_1+e_1)\oplus(w_2+e_{N_2})).$$

\end{theorem}

\begin{proof}
As remarked earlier, if we delete the edges of color $j\geq N_1+1$, we are left with $2^{N_2-k_2}$ disjoint copies of $A_1$, one for each vertex $v$ of $A_2$. Let $V_2$ denote the set of all vertices in $A_2$. It follows that if $1\leq j\leq N_1-1$, each $j/(j+1)$ colored face is contained in a copy of $A_1$ at a vertex $v\in V_2$ and, conversely, each $j/(j+1)$ colored face in such a copy will be a $j/(j+1)$ colored face in $X_1\otimes X_2$. Let us fix a vertex $v\in V_2$ and assume it is white for the time being. The transposition that switches the two adjacent white vertices of a $j/(j+1)$ colored face is given by
$$(c\oplus v,(c+c_j)\oplus v)=(c,c+c_j)\oplus v,$$ where $c\in C_1/\mathcal{C}_{k_1}$.
The product of all such transpositions over a set of orbit representatives for the action of $c_j$ on $C_1/\mathcal{C}_{k_1}$ is simply 
$$\rho_{j,1}\oplus v.$$
This element is the product of all transpositions swapping each white vertex with its opposite in any $j/(j+1)$ colored face, all this taking place in the copy of $A_1$ associated to the white vertex $v$. 

If $v$ is a black vertex, then the white vertices of the $j/(j+1)$ colored faces are tensor products of the black vertices in the $j/(j+1)$ colored faces of $A_1$ with $v$. 
The element $\rho_{j,1}'\oplus v$ is the product of all transpositions swapping each black vertex with its opposite in any $j/(j+1)$ colored face, all this taking place in the copy of $A_1$ associated to the black vertex $v$.
It follows that each $\rho_j$ is the product of the $\rho_j^{(v)}$ as given in the theorem statement. 

The monodromy generator $\rho_{N_1}$ is the product of the transpositions that swap the white vertices of the $N_1/(N_1+1)$ colored faces. Each such face has a unique vertex of the form $w_1\oplus w_2$, where the $w_i$ are white vertices of $A_i$. The adjacent white vertex is given by $(w_1+e_{N_1})\oplus(w_2+e_1)$.

It follows that
$$\rho_{N_1}=\prod_{\substack{w_1 \in W_1 \\ w_2 \in W_2}}(w_1\oplus w_2,(w_1+e_{N_1})\oplus(w_2\oplus e_1)).$$

That $\rho_j$ can be described as claimed for $N_1+1\leq j\leq N_1+N_2-1$ is demonstrated by complete analogy with the argument at the beginning of the proof. 

The $N$-th generator $\rho_N$ can be described as
$$\prod_{\substack{w_1 \in W_1 \\ w_2 \in W_2}}(w_1\oplus w_2,(w_1+e_1)\oplus(w_2+e_{N_2})).$$
We arrive at this description in the same way we found $\rho_{N_1}$.
\end{proof}

\begin{remark}
We could also describe this result in the labeling defined by the map $\gamma$ in Proposition \ref{prop:label}. While we defined this labeling via a map from the labeling used in Theorem \ref{Thm:monBN}, we could have defined the new labeling first. The proof that such a labeling exists without relying on our original description in terms of elements of $\mathbb{F}_2^N$ uses induction on the dimension of the cube. This approach makes it clear that a given $N$-cube can be obtained by taking tensor products of lower dimensional cubes. This is one of the most useful aspects of such a labeling and helps our understanding of the geometric meaning of the tensor product of Adinkras. 
\end{remark}



\begin{remark}
If $N_2=1$ then there is no $j$ such that $N_1+1\leq j\leq N-1=N_1$. Therefore, the generators of the monodromy group for the tensor product are completely determined by $\rho_{j,1}$ for $1\leq j\leq N_1$. Similarly, if $N_1=1$ then there is no $j$ such that $1\leq j\leq N_1-1=0$. We can choose $\rho_2,\ldots,\rho_{N_2+1}$ as  generators of the monodromy group of the tensor product, and these elements depend only on $\rho_{j,2}$. Note that if $N_1 = 1$, then $\rho_N$ contains the information specifying which copy of $A_1$ is connected to which by the new color. If $N_2 = 1$, then the same statement holds with $\rho_{N_1}$ and $A_2$ in place of $\rho_N$ and $A_1$.

The case of $N_1=N_2=1$ is unique in that,  for both $X_1$ and the tensor product $X_1\otimes X_1$, $\rho_j$ is the identity for all $j$. We can see this immediately for $X_1$ since it is the Belyi base $\bP^1(\bC)$ with embedded graph $\Sigma_0=A_1$. We have $X_1\otimes X_1=\bP^1(\bC)$ with embedded graph $A_1\otimes A_1$. Therefore the tensor product has $2$ faces, both with boundary $A_1\otimes A_1$. For $\rho_j$ to be non-trivial, one would need to differentiate between the $1/2$ and $2/1$ faces.
\end{remark}

We can now look at the action of the tensor product on the Riemann surfaces themselves. 
Let us first consider the case in which $N_1=N_2=1$.
We can interpret the tensor product $A_1\otimes A_1=A_2$ of Adinkras as an operation on Riemann surfaces: $X_1^{(a)}\otimes X^{(b)}_1=X_2$. We can view this operation as first taking two  copies of $X_1^{(a)}$ and two copies of $X^{(b)}_1$, then removing the $2$-face of each copy of $\bP^1(\bC)$, and then connecting their $1$-skeletons by adding two $2$-cells. Each $2$-cell is bounded by the four copies of $A_1$, alternating $a$ and $b$. They are connected with opposite orientations to create $\bP^1(\bC)$. The $4$ copies of $A_1$ can together be thought of as the equator. The fact that the $2$-cells are added with opposite orientations can be related to the different copies of $A_1$ being associated to white and black vertices; we will make this clearer in the general case. 

Next, let us  consider the case $N_1\neq 1$, $N_2=1$.

\begin{corollary}
If $X\in\mathcal{X}_{(N_1,k_1)}$ and $X_1$ is the unique element of $\mathcal{X}_{(1,0)}$, then the Riemann surface $X\otimes X_1$ can be viewed as a $2^{N_1-k_1-2}$-point connected sum of $X$ with itself. 
\end{corollary}
\begin{proof}
Let $A\in\mathcal{A}_{(N_1,k_1)}$ be the Adinkra out of which $X$ is constructed, and let $A_1$ be the Adinkra associated to $X_1$. There are two copies of $A$ in $A\otimes A_1$, corresponding to $w$ and $b$, the unique white vertex and black vertex of $A_1$.  If we ignore the edges of color $N_1+1$, as well as the faces incident to them in $X\otimes X_1$, we are left with two copies of $X$ with the $N_1/1$ faces removed. For $1\leq j\leq N_1-1$, we find that
$\rho_j^{(w)}=\rho_{j,1}\oplus w$, showing that the subsurface of $X\otimes X_1$ corresponding to the first copy of $A$ has all of the same $j/(j+1)$ faces as $X$ for $1\leq j\leq  N_1-1$. Since $N_1$ and $1$ are no longer adjacent, there are no $N_1/1$ colored faces. Therefore, $X^{(w)}$ is given by $X$ with the $N_1/1$ colored faces removed. Similarly, $X^{(b)}$ is a disjoint copy of $X$ with the $N_1/1$ colored faces removed. 
In summary, we are left with two copies of $X$, each with $2^{N_1-k_1-2}$ punctures; the boundary of each puncture is an $N_1/1$ colored loop.

Consider  a fixed puncture in $X^{(w)}$. 
The two white vertices in $X^{(w)}$ that bound the puncture are $c\oplus w$ and $(c+c_{N_1})\oplus w$ for some $c\in C_1/\mathcal{C}_{k_1}$. The  $2$-cycles in $\rho_{N_1}$ corresponding to these two vertices, namely 
$$(c\oplus w,(c+e_{N_1})\oplus b)\ \ \textrm{and}\ \left((c+c_{N_1})\oplus w,(c+e_1)\oplus b\right),$$
describe the two $N_1/(N_1+1)$ colored faces in $X$ that connect this puncture to the puncture on $X^{(b)}$  corresponding to $c\in C_1/\mathcal{C}_{k_1}$. Similarly, the two $2$-cycles in $\rho_N$ that correspond to $c$ are the two $N/1$ colored faces connecting the punctures. 
Therefore, we see that each hole of $X^{(w)}$ is connected to the corresponding hole of $X^{(b)}$ by a tube made of two $N_1/N$ colored faces and two $N/1$ colored faces. If we  consider a single hole and ignore the others, we  obtain the connected sum of $X$ with itself. Repeating for every puncture yields a $2^{N_1-k_1-2}$-point connected sum of $X$ with itself. 
\end{proof}

Note that for a connected sum of oriented manifolds, the orientation on the boundary of the holes that are connected is reversed. This is incorporated by the tensor product structure on the Adinkra. The holes that need to have their orientations reversed correspond to the black vertices of $X_1$.

If both  $N_1\neq 1$ and $N_2\neq 1$, the tensor product  ceases to be a connected sum since there are more than two copies of $A_1$. The $N_1/(N_1+1)$ faces and $N/1$ faces connect a puncture to two different copies of $A_1$, that is, the holes are no longer connected by tubes. In particular, the $N_1+1$ edge adjacent to an $N_1/(N_1+1)$ face bounds a puncture in a copy of $A_2$, and similarly for the $N$-colored edge adjacent to an $N/1$ face. For lack of a better term, we will refer to this as a multi-point connected sum between $A_1$ and $A_2$. We coin this terminology by analogy with the previous case, which gave rise to being a multi-point connected sum of $A$ with itself.


\begin{corollary}
If $X_1\in\mathcal{X}_{(N_1,k_1)}$ and $X_2\in\mathcal{X}_{(N_2,k_2)}$, then
the Riemann surface $X=X_1\otimes X_2$ can be viewed as a multi-point connected sum of $2^{N_2-k_2}$ copies of $X_1$ and $2^{N_1-k_1}$ copies of $X_2$. 
\end{corollary}

\begin{proof}
Forgetting the edges of colors $N_1+1,\dots, N$ and the $2$-cells connected to them in $X_1\otimes X_2$ leaves $2^{N_2-k_2}$ copies of $X_1$ with the $2^{N_1-k_1-2}$ $N_1/1$ colored faces removed. We  let $X_1^{(v)}$ denote the copy of $X_1$ attached to the vertex $v\in A_2$, and similarly we  let $X_2^{(v)}$ denote the copy of $X_2$ attached to the vertex $v\in A_1$. For $1\leq j\leq N_1-1$, $\rho_j^{(v)}$ describes the $j/(j+1)$ colored faces in $X_1^{(v)}$, showing that $X_1^{(v)}$ has all of the same $j/(j+1)$ faces as $X_1$ since $\rho_j^{(v)}$ can be obtained from $\rho_{j,1}$ as described in Theorem \ref{Thm:rhoTen} for $j\leq N_1-1$. This accounts for all of the faces that do not have an edge of color greater than $N_1$ bounding it, so that $X_1^{(v)}$ does not have any $N_1/1$ colored faces. 

Similarly, forgetting the edges of colors $1,\dots, N_1$ and the $2$-cells connected to them in $X_1\otimes X_2$ leaves $2^{N_1-k_1}$ copies of $X_2$ with the $N_2/1$ colored faces removed. Note that  $X_2^{(v)}$
and $X_1^{(v')}$  intersect only at vertices (if at all). We separate them at the vertices so that we can see, in the next step, how the vertices are forced to be equated. 

Fix a puncture in $X_1^{(v')}$. The two adjacent white vertices in the $N_1/(N_1+1)$ colored loop are given by $v\oplus v'$ and $(v+c_{N_1})\oplus v'$ for some vertex $v$ of the same color as $v'$. For simplicity of  exposition, we assume $v$ and $v'$ are both white.  The two $2$-cycles in $\rho_{N_1}$ corresponding to these vertices, namely
$$(v\oplus v', (v+e_{N_1})\oplus (v'+e_1)),\ \ ((v+c_{N_1})\oplus v', (v+e_1)\oplus (v'+e_1)),$$
describe the two $N_1/(N_1+1)$ colored faces that connect the puncture on $X_1^{(v')}$ to the corresponding puncture on $X_1^{(v'+e_1)}$. The edges of color $N_1$ that bound these faces correspond to edges in the boundaries of the copies of $X_1^{(v')}$ that are being connected. The other two edges of these faces have color $N_1-1$ and are in the boundaries of the copies of $X_2^{(v)}$. It follows  that the $N_1/(N_1+1)$ faces of the tensor product connect two copies of $X_1$ along part of the boundary of corresponding punctures in one direction and connects two copies of $X_2$ in the other direction. This identifies the vertices along the punctures of $X_1^{(v')}$ with the vertices along the attached holes of $X_2^{(v)}$. 

After the $N_1/(N_1+1)$ colored faces are added, the resulting manifold still has boundary. We must add the $N/1$ colored faces along that boundary according to $\rho_N$. Note that the $N/1$ faces connect holes of $X_1^{(v)}$ in one direction and holes of $X_2^{(v')}$ in the other, but will in general connect different holes than the $N_1/(N_1+1)$ colored faces. More precisely, we mean that if a fixed puncture in $X_1^{(v)}$ is connected to a hole in $X_1^{(v')}$, then the $N/1$ colored face will, in general, connect the original puncture to a puncture sitting inside yet another copy of $X_1$. 
\end{proof}

Theorem \ref{Thm:rhoTen} and its corollaries provide a new approach to the question of whether or  an Adinkra is factorizable, i.e., whether it can be written as the tensor product of two other Adinkras.

\section{The Belyi Curves Viewed Algebraically}
\label{sec:Algebraic}
In this section, we will describe the curves lying in $\mathcal{X}_{(N,k)}$ algebro-geometrically. First, we will demonstrate that these curves have Fuchsian uniformizations and that the uniformizing groups can be determined explicitly. Next, we will present the  curve $X_N$, associated to a hypercube, as a complete intersection of $N-3$ quadrics in $\mathbb{P}^{N-1}(\bC)$, and the curves $X\in\mathcal{X}_{(N,k)}$ as  quotients of $X_N$ via  groups of fixed-point-free automorphisms. The model for $X_N$ that we will introduce initially is visibly defined over the field $\mathbb{Q}(\zeta)$, where $\zeta$ is a primitive $2N$-th root of unity. This is not surprising because a Belyi curve is always definable over the algebraic closure $\overline{\mathbb{Q}}$ of $\mathbb{Q}$.  However, we will show that $X_N$ and the quotients $X\in\mathcal{X}_{(N,k)}$ are in fact definable over $\mathbb{Q}$. 

\subsection{Fuchsian Uniformizations}
It is a fact from covering space theory that every compact Riemann surface can be \emph{uniformized}, i.e.,  described as a quotient of  $\mathbb{C}$, $\mathbb{P}^1(\mathbb{C})$, or the upper half-plane $\mathbb{H}$ by a discrete subgroup of the corresponding automorphism group. In this section, we will describe how to uniformize the surfaces $X_{(N,k)}$. As we have seen, the theory behind the curves for $N\leq 3$ is not very interesting, so we will  discuss the story only for $N\geq 4$. For $N=4$, we will see that $X_4$ and its quotient by the unique doubly even code in $\mathbb{F}_2^4$  can be described as quotients of $\mathbb{P}^1(\mathbb{C})$, while for $N\geq 5$, the curves in $\mathcal{X}_{(N,k)}$ will be described as quotients of $\mathbb{H}$. More precisely, for any $N\geq 4$, we will determine the Fuchsian group $\Gamma$ that uniformizes $X_N$ and show that all of the curves in $\mathcal{X}_{(N,k)}$ correspond to subgroups of $\Gamma$. Thus, we can study the curves in $\mathcal{X}_{(N,k)}$ in a uniform way by studying subgroups of a fixed Fuchsian group.

In Section 3, we described the curves in $\mathcal{X}_{(N,k)}$ and their quotients explicitly in terms of monodromy data. This data, together with results in \cite{Carocca}, will give us the uniformization picture we are after. 
Let us begin by recalling how monodromy data gives rise to uniformizations. A good reference for the details of this discussion is \cite{Girondo:Text}. Suppose we have a branched covering map $f\colon X\to S$. The monodromy representation of $f$ is the map 
$$M_f\colon \pi_1(S-\{y_1,\dots, y_n\})\to S_d,$$ where $\{y_1,\dots,y_n\}$ is the branch locus of $f$, $d$ is the degree, and $S_d$ is the symmetric group on $d$ letters. The map $M_f$ describes how the fundamental group acts on the fiber at an unbranched value and how it depends on a choice of bijection between such a fiber and the set of $d$ elements; different choices result in corresponding conjugacies inside the symmetric group or the fundamental group. Now suppose that $S$ is uniformized by a Fuchsian group $\Gamma$. The monodromy data can be given in terms of a map $M_\Gamma\colon \Gamma\to S_d$ that fits into the following commutative diagram:

$$\xymatrix{\pi_1(S-\{y_1,\dots, y_N\})\ar^{M_f}[dr]\ar[d]_\rho &\\
\Gamma\ar_{M_\Gamma}[r]&S_d}$$

Here, the map $\rho$ sends a loop $\gamma$ to the unique transformation of the universal cover (the upper-half plane if $N\geq 5$) that sends a fixed choice of initial point for a lift of $\gamma$ to the then determined endpoint. The map $\rho$ is surjective, and since the kernel of $\rho$ is contained in the kernel of $M_f$, we can use $\rho$ to define $M_\Gamma$ in a well-defined manner. In this situation, the curve $X$ is uniformized by the group $M_\Gamma^{-1}(I(1))$, where $I(1)$ denotes the stabilizer of $1$. 

For all of the curves in $\mathcal{X}_{(N,k)}$, we have seen that the corresponding monodromy elements $\sigma_0$, $\sigma_1$, and $\sigma_\infty$ have orders $N$, $N$, and $2$, respectively. Moreover, in the language of \cite{Girondo:Text}, the Adinkras are \emph{uniform}, in the sense that each vertex has the same number of incident edges. It follows from results in \cite{Girondo:Text} that all of the curves in $\mathcal{X}_{(N,k)}$ are uniformized by normal torsion-free subgroups of the $(N,N,2)$-triangle group 

$$\Gamma_{N,N,2}=\gen{x_1,x_2,x_3| x_1^N=x_2^N=x_3^2=x_1x_2x_3=1}.$$

The group $\Gamma_{N,N,2}$ acts on the upper half-plane for $N\geq 5$,  on the standard plane for $N=4$, and  on the sphere for $N=1,2,3$. In all cases, we can realize $\Gamma_{N,N,2}$ as a group of transformations as follows. Let $R$ be a triangle with vertices $v_1,v_2,v_3$ and  angles $\pi/N,\pi/N,\pi/2$ at the respective vertices. Let $x_i$ be the transformation obtained by performing the two reflections in the edges containing $v_i$ in ascending order modulo $3$. Then $x_1$ is just rotation through an angle of $2\pi/N$ at $v_1$, $x_2$ is rotation through an angle of $2\pi/N$ at $v_2$, and $x_3$ is rotation by an angle of $\pi$ at $v_3$. If $D=R\cup R'$, where $R'$ is the image of $R$ under reflection in one of the edges, then $D$ is a fundamental domain for the action of the triangle group on the ambient space and the quotient is always the Riemann sphere.

 Our goal is to describe explicitly the groups uniformizing these curves. We will do this by first explicitly uniformizing the map from $B_N$ to the Belyi base and then using this uniformization to contstruct the other uniformizations. 

\begin{proposition}
The curve $B_N$ is uniformized by a Fuchsian group $\Gamma_N$ with presentation 
$$\Gamma_N=\gen{y_1,\dots,y_N|y_1^2=\cdots=y_N^2=y_1\cdots y_N=1}.$$
\end{proposition}
\begin{proof}

Let $\tilde{\beta}\colon B_N\to\mathbb{P}^1(\bC)$ be the usual map. From the above discussion, we see that the Belyi base is uniformized by $\Gamma_{N,N,2}$. Therefore, we need to consider the following commutative diagram:

$$\xymatrix{\pi_1(\mathbb{P}^1(\bC)-\{0,1, \infty\})\ar^{M_{\tilde{\beta}}}[dr]\ar[d]_\rho &\\
\Gamma_{N,N,2}\ar_{M_{\Gamma_{N,N,2}}}[r]&S_N}$$

The fundamental group $\pi_1(\mathbb{C}\mathbb{P}^1-\{0,1, \infty\})$ has a presentation given by $\gen{w_1,w_2,w_3|w_1w_2w_3=1}$ with $w_1$ corresponding to a loop around $0$, $w_2$ to a loop around $1$, and $w_3$ to a loop around $\infty$. By definition of the monodromy map, we have
$$M_{\tilde{\beta}}(w_i)=\sigma_i^{-1},$$ so that
$$M_{\Gamma_{N,N,2}}(x_i)=\sigma_i^{-1},$$
where $\sigma_i$ are the permutation representation pair for $\tilde{\beta}$. Note that the orders of the $\sigma_i$ are what make the map $M_{\Gamma_{N,N,2}}$ well defined. 
Since the action of the subgroup of $S_N$ generated by the $\sigma_i$ is simply transitive, it follows that the stabilizer is trivial. Therefore, by results found in \cite{Girondo:Text}, $B_N$ is uniformized by $\Gamma_N=\ker{M_{\Gamma_{N,N,2}}}\subseteq\Gamma_{N,N,2}$ and the map $\tilde{\beta}$ is simply the map 
$$\Delta/\Gamma_N\to \Delta/\Gamma_{N,N,2}$$ induced by inclusion, where $\Delta$ is the upper half-plane for $N\geq 5$ and the Euclidean plane for $N=4$. It is easy to see that the elements $y_i=x_1^ix_3x_1^{N-i}$ all lie in  $\ker{M_{\Gamma_{N,N,2}}}$. This can be seen group theoretically, but also by noting that these elements correspond to the $N$ lifts of the loop around $\infty$ which are the $N$ loops around the centers of the faces of $B_N$. Let $K$ be the subgroup of $\Gamma_N$ generated by the $y_i$. One can check that $K$ is a normal subgroup of $\Gamma_{N,N,2}$ and that $\Gamma_{N,N,2}/K\cong \mathbb{Z}/N\mathbb{Z}$. Since we know that $\Gamma_N$ has index $N$ in $\Gamma_{N,N,2}$, it follows that $K=\Gamma_N$, so that $\Gamma_N$ has the desired presentation. 

Note that $B_N$ has $N$ orbifold points of order $2$, corresponding to the centers of the faces. This is seen in the uniformization picture by noting that $\Gamma_N$ is \emph{not} a torsion-free subgroup. 
\end{proof}

\begin{proposition}
\label{uniform}
The curve $X_N$ is uniformized by the torsion-free normal subgroup $\Gamma_{(N,0)}=\Gamma_N'$, the commutator subgroup of $\Gamma_N$. More concretely, $\Gamma_{(N,0)}$ is the normal closure of the set
$$\{(y_iy_j)^2|1\leq i,j\leq N\}.$$
\end{proposition}

\begin{proof}
The proof of this theorem is a slight modification of the proof of Theorem 2.3 in \cite{Carocca} and we review it here in a bit more detail than is given in the reference. We consider the  commutative diagram

$$\xymatrix{\pi_1(B_N-\{\zeta_j \})\ar^{M_{f}}[dr]\ar[d]_\rho &\\
\Gamma_N\ar_{M_{\Gamma_N}}[r]&S_{2^{N-1}}}$$
where $\zeta_1=e^\frac{\pi i}{N}$,$\zeta_j=\zeta_1^{2j-1}$, and $f$ is the usual map $f\colon X_N\to B_N$. Note that the degree of this map is $2^{N-1}$.

We labeled the monodromy generators of  $f$,  $\rho_j$, by the white vertices in $X_N$. 
The fundamental group $\pi_1(B_N-\{\zeta_j \})$ is equal to $\gen{z_1,\dots, z_N|z_1\cdots z_N=1}$, with $z_j$ corresponding to the loop around  $\zeta_j$. We have
$$M_f(z_j)=\rho_j$$ and
$$M_{\Gamma_N}(y_j)=\rho_j.$$

The stabilizer of a white vertex is trivial, so the kernel of $M_{\Gamma_N}$ uniformizes $X_N$ and the map $f\colon X_N\to B_N$ is given in the uniformization picture as 
$$\Delta/\Gamma_{(N,0)}\to \Delta/\Gamma_N,$$
where $\Delta$ is as above and we have wrote $\Gamma_{(N,0)}$ now instead of $\ker{M_{\Gamma_N}}$. Let $F$ be the normal closure in $\Gamma_N$ of the set 
$$\{(y_iy_j)^2|1\leq i,j\leq N\}.$$
Then it is easy to see that $\Gamma_N/F\cong(\mathbb{Z}/2\mathbb{Z})^{N-1}$, from which it follows by index considerations that $F=\Gamma_{(N,0)}$, so that $\Gamma_{(N,0)}$ has the claimed presentation.

The Riemann surface $X_N$ has no orbifold points, which shows that the  group  $\Gamma_{(N,0)}$ is torsion-free.

\end{proof}

Let us consider once more the commutative diagram used in the proof of Proposition \ref{uniform}. The loop $z_i$ corresponds under $\rho$ to the transformation that sends a white vertex $w$ to the other white vertex that makes up an $i/(i+1)$ colored face. Motivated by this, let us define a group homomorphism 
\begin{eqnarray}
\varphi\colon \Gamma_N &\to &\mathbb{F}_2^N\nonumber\\
y_i&\mapsto& c_i\nonumber
\end{eqnarray}
where, as usual, $c_i$ is the $i$-th generator of the maximal even code. The map $\ph$ is surjective onto the maximal even code $C_N$, so 
$$\Gamma_N/\ker{\varphi}\cong C_N\cong(\mathbb{Z}/2\mathbb{Z})^{N-1}.$$

Moreover, it is clear that $\ker{\ph}=\ker{M_{\Gamma_N}}=\Gamma_{(N,0)}$. Therefore, the group uniformizing $X_N$ is described group theoretically as the kernel of a map from $\Gamma_N$ to $C_N$, namely the map $\ph$. Now let $\mathcal{C}_k$ be a $k$-dimensional doubly even subcode. Then $$\Gamma_N/\varphi^{-1}(\mathcal{C}_k)\cong C_N/\mathcal{C}_k$$
With a little more work, we conclude in the proposition below, that $\varphi^{-1}(\mathcal{C}_k)$ uniformizes the curve in $\mathcal{X}_{(N,k)}$ corresponding to $\mathcal{C}_k$.  

\begin{proposition}
Let $\ph\colon \Gamma_N\to C_N$ be as above. Then $\Gamma_{(N,k)}=\ph^{-1}(\mathcal{C}_k)$ uniformizes the curve $X\in\mathcal{X}_{(N,k)}$ associated to $\mathcal{C}_k$. Suppose further that $\mathcal{C}_k$ is generated by $\{v_1,\dots, v_k\}$, with each $v_i$ written as $\sum_{j=1}^{N-1}a_{ij}c_j$ for uniquely determined $a_{ij}\in\{0,1\}$. If we set  $$y_{v_i}=\prod_{j=1}^{N-1}y_i^{a_{ij}}$$
then $\ph^{-1}(\mathcal{C}_k)$ can be described as the normal closure in $\Gamma_N$ of the set
$$\{y_{v_i}|i=1,\dots, k\}\cup\{(y_iy_j)^2|1\leq i,j\leq N\}.$$
\end{proposition}

\begin{proof}
Consider the commutative diagram
$$\xymatrix{\pi_1(B_N-\{\zeta_j \})\ar^{M_{X}}[dr]\ar[d]_\rho &\\
\Gamma_N\ar_{M_{\Gamma_N}}[r]&S_{2^{N-k-1}}}$$
where $M_{X}$ is defined in terms of the monodromy data for $f_X\colon X\to B_N$ and $M_{\Gamma_N}$ is defined to make the diagram commute. Then, as usual, $X$ is uniformized by the kernel of $M_{\Gamma_N}$. Since two white vertices are identified exactly when they differ by an element of $\mathcal{C}_k$, it follows at once that this kernel agrees with $\ph^{-1}(\mathcal{C}_k)$, whence $\ph^{-1}(\mathcal{C}_k)$ uniformizes $X$ as desired. 

Now let $K$ be the normal closure of 
$$\{y_{v_i}|i=1,\dots, k\}\cup\{(y_iy_j)^2|1\leq i,j\leq N\}.$$
The $y_{v_i}$ were chosen such that $\ph(y_{v_i})=v_i$. Therefore, each element of $\mathcal{C}_k$ can be written  as the image of  a product of the elements $y_{v_i}$. So, if  $z\in\varphi^{-1}(\mathcal{C}_k)$, then the element $\varphi(z)$ can also be written as $\varphi(y)$ where $y$ is a product of elements in the set $\{y_{v_i}\}$. It follows that $zy^{-1}\in\ker{\varphi}$, from which we deduce that $z=y w$ for some element $w\in\ker{\varphi}$. Since the kernel of $\varphi$ has been determined to be the normal closure of $\{(y_iy_j)^2|1\leq i,j\leq N\}$, it follows that $z$ can be written as desired. 

\end{proof}

\subsection{An Algebraic Model for the Hypercube Surfaces}
We demonstrated above that $X_N$ contains a subgroup of automorphisms $H\cong(\mathbb{Z}/2\mathbb{Z})^{N-1}$ and that the quotient $X_N/H\cong B_N$. The curve $B_N$ is the orbifold of signature $(0,N;2,\dots, 2)$; indeed, it has $N$ order $2$ points at the $N$-roots of $-1$. Such a curve is called a \emph{generalized Humbert curve} in the language of \cite{Carocca}. It is shown in \cite{Carocca} that any generalized Humbert curve $S$ such that $S/H$ has signature $(0,N;2,\dots, 2)$ has a model of the form $C(\lambda_1,\dots,\lambda_{N-3})$ where $C(\lambda_1,\dots,\lambda_{N-3})$ is given by the zero-locus of the following equations
$$\left\{\begin{array}{ccc} x_1^2+x_2^2+x_3^2&=&0\\
\lambda_1 x_1^2+x_2^2+x_4^2&=&0\\
\vdots &\vdots &\vdots\\
\lambda_{N-3}x_1^2+x_2^2+x_{N}^2&=&0
\end{array}\right.$$
and $\lambda_i\in\mathbb{C}-\{0,1\}$. The curve $S$ comes equipped with the degree $2^{N-1}$ map 
\begin{eqnarray}
\pi\colon S&\to &\mathbb{C}\nonumber\\
x&\mapsto & \frac{x_2^2}{x_1^2},\nonumber
\end{eqnarray}
whose branch locus is the set $\{0,-1,\infty,-\lambda_1,\dots,-\lambda_{N-3}\}$, as well as with the group $H\cong(\mathbb{Z}/2\mathbb{Z})^{N-1}$ of deck transformations of $\pi$ generated by the maps
\begin{eqnarray}
a_j\colon S&\to &S\nonumber\\
\left[x_1\colon \cdots \colon x_j \colon \cdots \colon x_N\right]&\mapsto &\left[x_1\colon \cdots \colon -x_j \colon \cdots \colon x_N\right].\nonumber
\end{eqnarray}

The map $X_N\to B_N$ has the $N$-th roots of $-1$ as branch locus. If we let $f\colon B_N\to\mathbb{P}^1(\mathbb{C})$ be a M\"{o}bius transformation that maps the $N$-th roots of $-1$ into $\{0,-1,\infty,-\lambda_1,\dots,-\lambda_N\}$, then Theorem 4.3 of \cite{Carocca} states that $X_N$ is conformally equivalent to $C(\lambda_1,\dots,\lambda_N)$. 

Let $\zeta=e^\frac{\pi i}{N}$ and $\xi=\frac{\zeta^3-\zeta^{-1}}{\zeta-\zeta^3}$. The M\"{o}bius transformation 
$$f(z)=\frac{z-\zeta}{z-\zeta^{-1}}\xi$$
maps the points $\zeta,\zeta^3,\zeta^{2N-1}$ onto $0,-1,\infty$, respectively.  The image under $f$ of the other roots of unity will necessarily lie on the negative real axis. 
Let us order the $N$-th roots of $-1$ by setting
$$\zeta_i=\zeta^{2i-1}.$$
If we set $\mu_i=f(\zeta_i)$ for $i=1,\dots, N$, then $X_N$ is conformally equivalent to $C(-\mu_2,\dots,-\mu_{N-1})$. Note that $\mu_1=0,\mu_2=-1$, and $\mu_N=\infty$. From now on, we will use $X_N^{\mathrm{alg}}$ to denote this model and $B_N^{\mathrm{alg}}$ to denote the target of the map $\pi$.  

Recall that the belyi map for $B_N$ is given by $\tilde{\beta}\colon B_N\to\mathbb{C}$, where  
$$\tilde{\beta}(x)=\frac{x^N}{x^N+1}.$$
Therefore, the Belyi map for $X_N^{\mathrm{alg}}$ is given by $\beta^{\mathrm{alg}}=\tilde{\beta}\circ f^{-1}\circ \pi\colon X_N\to\mathbb{C}$. 

Let us describe the vertices of the $N$-cube sitting inside $X_N^{\mathrm{alg}}$. If we set $\zeta_0=f(0)$ and $\zeta_\infty=f(\infty)$, then the white vertices are the $2^{N-1}$ points of $\pi^{-1}(\zeta_0)$ and the black vertices are the $2^{N-1}$ points of $\pi^{-1}(\zeta_\infty)$. Explicitly, we have
$$\pi^{-1}(\zeta_0)=[1\colon \pm\sqrt{\zeta_0}\colon\pm\sqrt{-1-\zeta_0}\colon\pm\sqrt{\mu_3-\zeta_0}\colon\cdots\colon\pm\sqrt{\mu_{N-1}-\zeta_0}]$$

and

$$\pi^{-1}(\zeta_\infty)=[1\colon\pm\sqrt{\zeta_\infty}\colon\pm\sqrt{-1-\zeta_\infty}\colon\pm\sqrt{\mu_3-\zeta_0}\colon\cdots\colon\pm\sqrt{\mu_{N-1}-\zeta_\infty}].$$

Next, let us describe the edges of the $N$-cube in $X_N^{\mathrm{alg}}$. Suppose that $B_N$ has the usual rainbow $(1,\dots, N)$, where the $i$-th color corresponds to the ray with argument $2\pi/N$. Then, $B_N^{\mathrm{alg}}$ inherits the rainbow $(1,\dots, N)$ with the edge of color $i$ corresponding to a circular arc which joins $\zeta_0$ to $\zeta_\infty$ and which crosses the negative real axis between $\mu_i$ and $\mu_{i+1}$ if $1\leq i\leq N-1$ and  crosses the positive real axis if $i=N$. The curve $X_N^{\mathrm{alg}}$ is given the same rainbow with the edges of color $i$ being given by the pre-images of the edge of color $i$ on $B_N^{\mathrm{alg}}$. 

The centers of the faces of $X_N^{\mathrm{alg}}$ are given by the vanishing of coordinate functions. Indeed, they are the pre-images of the branch points. Therefore, the centers of the faces are given by
$$\pi^{-1}(\mu_i)=\{x_{i+1}=0\}\cap X_{N}^{\mathrm{alg}},$$
where $i+1$ is computed modulo $N$, with the set of representatives taken as $\{1,\dots,N\}$.
In light of our description of the edges, we see that the points of $\pi^{-1}(\mu_{i+1})$ are the centers of the $i/(i+1)$ faces. In summary, we have proved the following. 

\begin{proposition}
Let $\zeta=e^\frac{\pi i}{N}$, $\xi=\frac{\zeta^3-\zeta^{-1}}{\zeta-\zeta^3}$, and $f(z)=\frac{z-\zeta}{z-\zeta^{-1}}\xi$. 
Set $\mu_i=f(\zeta^{2i-1})$ for $1\leq i\leq N$, and let $X_N^{\mathrm{alg}}=C(-\mu_3,\dots,-\mu_{N-1})\subseteq \mathbb{P}^{N-1}(\mathbb{C})$. Further, let $\beta^{\mathrm{alg}}=\tilde{\beta}\circ f^{-1}\circ \pi$. Then the Belyi pair $(X_N,\beta)$ is equivalent to the Belyi pair $(X_N^{\mathrm{alg}},\beta^{\mathrm{alg}})$. If $\zeta_0=f(0)$ and $\zeta_\infty=f(\infty)$, then the white vertices of $X_N^{\mathrm{alg}}$ are the points in $\pi^{-1}(\zeta_0)$ and the black vertices are the points in $\pi^{-1}(\zeta_\infty)$. If we give $X_N^{\mathrm{alg}}$ the rainbow that comes naturally from $B_N$, then the centers of the $i/(i+1)$ colored faces are given by $\pi^{-1}(\mu_{i+1})=\{x_{i+2}=0\}\cap X_N^{\mathrm{alg}}$.
\end{proposition}

The deck transformation group of $\pi$ is the group $H\cong(\mathbb{Z}/2\mathbb{Z})^{N-1}$ generated by the maps $a_j$ that switch the sign of the $j$-th coordinate. We claim that the group of deck transformations of $\beta^{\mathrm{alg}}$ is the semi-direct product $H\rtimes \mathbb{Z}/N\mathbb{Z}$. Indeed, we have the following proposition, of which the claim is a corollary. 

\begin{proposition}\label{RotProp}
Let $r\colon B_N^{\mathrm{alg}}\to B_N^{\mathrm{alg}}$ be the M\"{o}bius transformation that corresponds to the rotation of $B_N$ through an angle of $2\pi /N$ in the positive direction. Let $s\colon\mathbb{C}\mathbb{P}^{N-1}\to\mathbb{C}\mathbb{P}^{N-1}$ be the automorphism given by 
$$s[x_1\colon\cdots\colon x_N]=[x_N\colon x_1\colon\cdots\colon x_{N-1}]$$
and let $d\colon \mathbb{C}\mathbb{P}^{N-1}\to\mathbb{C}\mathbb{P}^{N-1}$ be the diagonal automorphism given by
$$d[x_1\colon\cdots\colon x_N]=[x_1,c_2x_2,\dots,c_Nx_N],$$
where $c_2=\sqrt{-\mu_{N-1}}$, $c_3=1$, and $c_i=\sqrt{\mu_{i-1}}$ for $i=4,\dots,N$. 
Set $\Theta=d\circ s|_{X_N^{\mathrm{alg}}}$. Then $\Theta$ is an automorphism of $X_N^{\mathrm{alg}}$ that is a lift of $r$. In particular, $\Theta$ is a deck transformation of $\beta^{\mathrm{alg}}$. 

\end{proposition}

The proof of this proposition is  in  the Appendix. 

\begin{corollary}\label{DeckTrafo}
The  $G$ of deck transformations of $\beta^{\mathrm{alg}}$ is generated by $H$ and $\Theta$. Moreover, there is an isomorphism of groups $$G\cong H\rtimes\mathbb{Z}/N\mathbb{Z}.$$ 
\end{corollary}
\begin{proof}
By inspection, the automorphisms $h\Theta^i$, where $h\in H$ and $0\leq i\leq i-1$, are all distinct. Therefore, the group generated by $\Theta$ and $H$ has order at least $2^{N-1}N$. On the other hand, the degree of $\beta^{\mathrm{alg}}$ is $2^{N-1}N$, so the full deck transformation group must be generated by $\Theta$ and $H$. 

Consider the natural map 
$$\gen{\Theta}\to G/H.$$
It is injective because all non-trivial powers of $\Theta$ permute the coordinates non-trivially. Further, it is an isomorphism, since both groups have order $N$. It then follows that there is a group isomorphism
$$G\cong H\rtimes\gen{\Theta}\cong H\rtimes \mathbb{Z}/N\mathbb{Z}.$$
\end{proof}

\subsection{The Quotients of Hypercube surfaces}
Our next task is to describe the curves in $\mathcal{X}_{(N,k)}$ algebraically. In order to do this, we will describe how a doubly even code gives rise to a subgroup of $H$ that acts fixed-point free on $X_N^{\mathrm{alg}}$. The quotient of $X_{N}^{\mathrm{alg}}$ by this subgroup is the algebraic description we are after. 

The curve $X\in\mathcal{X}_{(N,k)}$ is obtained from $X_N$ by identifying certain vertices and edges determined by a doubly even code $\mathcal{C}_k\subseteq\mathbb{F}_2^N$. We determine here how this identification translates to the algebraic picture. Recall that the white vertices of the $N$-cube sitting inside $X_N^\mathrm{\mathrm{alg}}$ are given by

$$\pi^{-1}(\zeta_0)=[1\colon \pm\sqrt{\zeta_0}\colon\pm\sqrt{-1-\zeta_0}\colon\cdots\colon\pm\sqrt{\mu_{N-1}-\zeta_0}]$$
and the black vertices by 
$$\pi^{-1}(\zeta_\infty)=[1\colon\pm\sqrt{\zeta_\infty}\colon\pm\sqrt{-1-\zeta_\infty}\colon\cdots\colon\pm\sqrt{\mu_{N-1}-\zeta_\infty}].$$

Each edge of $X_N^{\mathrm{alg}}$ arises from the analytic continuation of the function
$$z\mapsto [1\colon \sqrt{z}\colon\sqrt{-1-z}\colon\cdots\colon \sqrt{-\lambda_{N-3}-z}]$$
along an edge of $B_N^{\mathrm{alg}}$ that starts at $\zeta_0$ and ends at $\zeta_\infty$. Let us label the white vertices $w(+,\pm,\cdots,\pm)$ and the black vertices $b(+,\pm,\cdots,\pm),$ where the $i$-th sign is the sign of the $i$-th coordinate of the vertex. Note that since we are currently working in the affine chart $\{x_1=1\}$, the first sign will always be $+$. If we start at a given white vertex $w$ and travel along the edge of color $i$, we will end up at some black vertex $b(+,\pm,\dots,\pm)$. Which vertex we arrive at can be encoded by a sequence of signs that indicates whether or not the $i$-th coordinate of $b$ will have the same sign as that of $w$. We can determine whether or not the sign will change by examining the branch cuts that are needed to define $\sqrt{z}$. 

Let us now be precise about how we define the square root function. Let $\log(z)$ denote the principal branch of the logarithm obtained by making a branch cut along the negative real axis, with $\arg(z)\in (-\pi,\pi]$. We will take $\sqrt{z}$ to mean $e^{\frac{1}{2}\log(z)}$; crossing the branch cut corresponds to choosing the other branch of square root, which amounts to a switching of signs. The maps
\begin{eqnarray}
z&\mapsto &\mu_i-z,\ i\geq 2,\nonumber
\end{eqnarray}
send the negative real axis onto the ray $[\mu_i,\infty)$. We start at a white vertex $w$ and travel along the edge of color $i$. As we do so, we observe which branch cuts the edge of color $i$ on $B_N^{\mathrm{alg}}$ crosses; each branch cut it crosses corresponds to a sign switch. For example, the edge of color $1$ on $B_N^{\mathrm{alg}}$ crosses the negative real axis between $0$ and $-1$. Therefore, every branch cut is crossed and all of the signs will change. In particular, the vertex $w(+,\dots,+)$ is connected to $b(+,-,\dots,-)$ by the edge of color $1$. The edge of color $2$ will cross every branch cut except the one corresponding to $[-1,\infty)$, so every coordinate switches its sign except for the third coordinate. Each time we move down the rainbow one more color, we gain a coordinate whose sign does not change until we hit color $N$. The edge of color $N$ crosses the \emph{positive} real axis, so that all of the coordinates switch sign except for the second. We can summarize the above in the following table.

\begin{center}
\begin{tabular}{|c|l|}
\hline
{Color }
&
{Sign change}
\\
\hline
$1$&$(+,-,-,\cdots,-)$\\
\hline
$2$& $(+,-,+,-,\cdots,-)$\\
\hline
$3$& $(+,-,+,+,-,\cdots,-)$\\
\hline
$4$& $(+,-,+,+,+,-,\cdots,-)$\\
\hline
\vdots & \vdots\\
\hline
$N-1$& $(+,-,+,\cdots,+)$\\
\hline
$N$&$(+,+,-,-,\cdots,-)$\\
\hline
\end{tabular}
\end{center} 

Recall that in the standard picture of the $N$-cube, the white vertices correspond to even elements of $\mathbb{F}_2^N$ and the black vertices to odd elements. Further, the $i$-th color corresponds to $e_i$, the $i$-th standard basis vector, and two vertices are adjacent via color $i$ if they differ by $e_i$. Therefore, we would like to associate $e_i$ to the $N$-tuple of sign changes corresponding to the $i$-th color in the table above. Rather than continue working with these $N$-tuples of sign changes, we will associate to each of these sign changes the element of $H$ that acts as the corresponding sign changes. Putting everything together, we may define a map of $\mathbb{F}_2$-vector spaces 

\begin{eqnarray}
\psi\colon \mathbb{F}_2^N&\to & H\nonumber
\end{eqnarray}
by the rule 
$$\psi(e_j)=\left\{\begin{array}{lll}a_2\prod_{i=j+2}^Na_i&\textrm{if}&1\leq j\leq N-2\\
a_2&\textrm{if}&j=N-1\\
\prod_{i=3}^Na_i&\textrm{if}&j=N\end{array}\right.$$

We have therefore established a way to think of the colors of the $N$-cube as elements of $H$. Note that since each color takes us from a white vertex to a black vertex, it does not make sense geometrically to view the colors themselves as elements of $H$. However, if we restrict ourselves to only the even elements, then it \emph{does} make sense to view them as elements of $H$. 

\begin{theorem}\label{Dictionary}Let $C_N\subseteq\mathbb{F}_2^N$ be the maximal even sub code. Then the restriction of $\psi$ to $C_N$ is an isomorphism onto $H$. 
If $\mathcal{C}_k\subseteq C_N$ is a doubly even code, then $\psi(\mathcal{C}_k)\subseteq H$ is a subgroup of fixed-point-free automorphisms. If $X\in\mathcal{X}_{(N,k)}$ is the curve associated to the code $\mathcal{C}_k$, then $X\cong X_{N}^{\mathrm{alg}}/\psi(\mathcal{C}_k)$. 
\end{theorem}

\begin{proof}
First, let us show that $\psi$ is surjective. We note that
$$\psi(e_{j}+e_{j+1})=a_{j+2},\ 1\leq j\leq N-2.$$
It follows that $a_3,\dots, a_N$ are in the image of $\psi$. Since $\psi(e_{N-1})=a_2$ and $H$ is generated by $a_2,\dots, a_N$, we conclude that $\psi$ is surjective. By counting dimensions, we see that the $\ker\psi$ is $1$-dimensional. Since $e_1+e_{N-1}+e_N\in\ker\psi$ by inspection, we conclude that
$$\ker\psi=\gen{e_1+e_{N-1}+e_N}.$$

We now argue that the restriction of $\psi$ to $C_N$ is an isomorphism. First, since $a_2=\psi(e_1+e_N)$ and the other $a_i$ are the images of even codewords, we conclude that the restriction of $\psi$ to $C_N$ is surjective. Since $C_N$ and $H$ both have dimension $N-1$ as $\mathbb{F}_2$-vector spaces, it follows that the restriction of $\psi$ to $C_N$ is in fact an isomorphism. 

Not all of the automorphisms in $H$ act fixed-point free on $H$. In fact, any point of $X_N^{\mathrm{alg}}$ that has zero as its $i$-th coordinate will be a fixed point of $a_i$ and, conversely, any fixed point of an element of $H$ must have  zero as some coordinate. However, we have already demonstrated that such a point is the center of some face of $X_N^{\mathrm{alg}}$. It follows that all of the $a_i$ for $i=2,\dots,N$ have fixed points but  any element of $H$ that involves at least $2$ of the $a_i$ must act fixed-point free. We have already shown that the pre-image under $\psi$ of each $e_i$ is a codeword of weight $2$. Since the restriction of $\psi$ to $C_N$ is an isomorphism, it follows that no doubly-even code word can map to any $a_i$ for $i=1,\dots, N$. Therefore, the image of a doubly even codeword will act fixed-point free on $X_N^{\mathrm{alg}}$. 

By design, the automorphisms of $H=\psi(C_N)$ identify vertices and edges in the same way that the vertices and 
edges of the $N$-cube are identified. Since $\mathcal{C}_k\subseteq C_N$ is a subcode, $\psi(\mathcal{C}_k)$ makes the same identifications as $\mathcal{C}_k$ does on the $N$-cube, so
$$X\cong X_N^{\mathrm{alg}}/\psi(\mathcal{C}_k).$$
\end{proof}

\subsection{Monodromies} Theorem \ref{Dictionary} gives us an explicit way to view the action of the colors on the vertices and edges of the Adinkra sitting inside of $X_N^{\mathrm{alg}}$. We can use this to help us label the monodromies of the map $\beta^{\mathrm{alg}}\colon X_N^{\mathrm{alg}}\to\mathbb{C}$, as well as the monodromies of the maps $X^{\mathrm{alg}}\to B_N^{\mathrm{alg}}$ for $X\in\mathcal{X}_{(N,k)}$.  

We begin by describing the permutation pair associated to the Belyi map $\beta^{\mathrm{alg}}$. 
For $h \in H$, let $w_h$ be the white vertex described by the sign changes encoded by $h$. For example, 
$$w_1=[1\colon \sqrt{\zeta_0}\colon\sqrt{-1-\zeta_0}\colon\cdots\colon\sqrt{\mu_{N-1}-\zeta_0}]$$
and 
$$b_1=[1\colon \sqrt{\zeta_\infty}\colon\sqrt{-1-\zeta_\infty}\colon\cdots\colon\sqrt{\mu_{N-1}-\zeta_\infty}]$$
The map $\psi$ appearing in Theorem \ref{Dictionary} was designed in such a way that
$w_h$ will be connected to $b_{h\psi(e_i)}$ by the edge of color $i$. Let $\Theta$ be the rotational automorphism of $X_N^{\mathrm{alg}}$ that was described explicitly in Proposition \ref{RotProp}. One can check that $\Theta(w_1)=w_1$. Let $e_1$ denote the edge of color $i$ incident to $w_1$. Then
$$(e_1,\Theta\cdot e_1,\dots,\Theta^{N-1}\cdot e_1)$$
lists all the edges incident to $w_1$ in the order of the rainbow. 
If $w_h=hw_1$ is another white vertex, then the incident edges in the order of the rainbow are given by
$$(h\cdot e_1,h\Theta\cdot e_1,\dots,h\Theta^{N-1}\cdot e_1).$$

Therefore, let us label the edges of the Adinkra by elements of $G$, the full deck transformation group. Every $g\in G$ can be written uniquely as $h\Theta^i$ for some $i$;  by the above discussion, the element $h\Theta^i$ corresponds to the edge of color $i-1$ incident to $w_h$. 

We can therefore describe $\sigma_0$ as an element of $S_G$, the symmetric group on the elements of $G$:

$$\sigma_0=\prod_{h\in H}(h,h\Theta,h\Theta^2,\dots,h\Theta^{N-1})\in S_{G}.$$

Now let us describe $\sigma_1$. The black vertex $b_h$ is connected to $w_{\psi(e_i)h}$ by the color $i$. On the other hand, we have already agreed to assign the label $h\psi(e_i)\Theta^{i-1} \in G$ to the edge of color $i$ incident to $w_{h\psi(e_i)}$.  

It follows that 
$$\sigma_1=\prod_{h\in H}(\psi(e_N)h\Theta^{N-1},\psi(e_{N-1})h\Theta^{N-2},\dots,\psi(e_2)h\Theta,\psi(e_1)h)\in S_G.$$

The product $\sigma_1\sigma_0$ is
$$\sigma_{\infty}=\prod_{i=1}^N\prod_{h\in T_i}(h\Theta^{i-1},\psi(e_{i+1})\psi(e_{i})h\Theta^{i-1}),$$
where $T_i$ is a set of coset representatives for $H/\gen{\psi(e_{i+1})\psi(e_{i})}$. Notice that $h\Theta^{i-1}$ is the edge of color $i$ incident to $w_h$ and $\psi(e_{i+1})\psi(e_{i})h\Theta^i$ is the edge of color $i$ incident to $w_{\psi(e_{i+1})\psi(e_{i})h},$
so that the transposition lists the opposite edges of the $i/(i+1)$ colored face.

Lastly, the element $\pi_\infty$  is the product of the $4$-cycles that list the edges of each face as we move clockwise around the center. Using our labeling, we find that
$$\pi_\infty=\prod_{i=1}^N \prod_{h\in T_i}(h\Theta^{i-1},h\Theta^{i},\psi(e_{i})\psi(e_{i+1})h\Theta^{i-1},\psi(e_{i})\psi(e_{i+1})h\Theta^{i}).$$

Notice that we can obtain $\sigma_\infty$ from $\pi_\infty$ by omitting the second and fourth entries above, which corresponds to  looking only at the two edges of the same color in each face. This is in line with what we observed earlier when we were discussing monodromies. In summary, we have used our map $\psi$ and our explicit determination of the deck transformation group $G$ to label the monodromies $\sigma_0,\sigma_1,\sigma_\infty,$ and $\pi_\infty$ as elements of $S_G$.

Now, we  would like to describe the monodromies over each face on the curve $B_N^{\mathrm{alg}}$. Consider the usual map $\pi : X_N^\mathrm{\mathrm{alg}} \to B_N^\mathrm{\mathrm{alg}}$. Since $\pi$ is unramified at the vertices, we can label these monodromies by the vertices. Sticking to the notation  above, we  let $w_h$ be the white vertex corresponding to the element $h\in H$. Let $\rho_i$ be the monodromy corresponding to the $i/(i+1)$ colored face. Then $\rho_i$ is the product of the transpositions that interchange opposite white vertices in each face. If $w_h$ is one of these vertices, then $w_{\psi(e_i)\psi(e_{i+1})h}$ is the other. Therefore, if we identify $w_h$ with $h\in H$, then the  elements 
$$\rho_i=\prod_{h\in H_i}(h,\psi(e_i)\psi(e_{i+1})h)$$
describe the monodromies over the $i/(i+1)$ colored face, where $H_i$ is a set of coset representatives for $H/\gen{\psi(e_i)\psi(e_{i+1})}$. Therefore, one way to view $\rho_i$ is as follows. Write $H$ as a disjoint union of the cosets of $\gen{\psi(e_i)\psi(e_{i+1})}$ in $H$. Then $\rho_i$ is the product of the transpositions that interchange the two elements of each coset. 

The fact that
$\rho_1\cdots\rho_N=1$ is easily seen from this point of view since 
$$\prod_{i=1}^N\psi(e_i)\psi(e_{i+1})=\prod_{i=1}^{N+1}\psi(e_i)=1.$$

Now let $\mathcal{C}_k$ be a $k$-dimensional doubly even code, so that $\psi(\mathcal{C}_k)\subseteq H$ is a subgroup of fixed-point-free automorphisms of $X_N^{\mathrm{alg}}$, the quotient by which is $X^{\mathrm{alg}}$. Two white vertices $w_{h_1}$ and $w_{h_2}$ in $X_N^{\mathrm{alg}}$ are identified in $X^{\mathrm{alg}}$ exactly when $h_1=ch_2$ for some $c\in \mathcal{C}_k$. Equivalently, they are identified if $h_1$ and $h_2$ are in the same coset of $H/\mathcal{C}_k$, where we are identifying $\mathcal{C}_k$ with its image in $H$ via $\psi$. We may therefore label the white vertices of $X^{\mathrm{alg}}$ with elements of $H/\mathcal{C}_K$. If $\rho_i(X^{\mathrm{alg}})$ is the monodromy at the $i/(i+1)$ colored face, then 
$$\rho_i(X^{\mathrm{alg}})=\prod_{h\in H_i}(h,\psi(e_i)\psi(e_{i+1})),$$
where $H_i$ is a set of coset representatives for $\gen{\psi(e_{i+1})(\psi(e_{i})}$ in $H/\mathcal{C}_k$. In fact, since $C_k\cap\gen{\psi(e_i)\psi(e_{i+1})}=\{1\}$, we could also define $\rho_i$ by

$$\rho_i(X^{\mathrm{alg}})=\prod_{t\in T_i}(t,\psi(e_i)\psi(e_{i+1})),$$
where $T_i$ is a set of representatives for the $(\gen{\psi(e_i)\psi(e_{i+1})},\mathcal{C}_k)$ double cosets in $H$. 

We have therefore expressed the monodromies of $\pi$ as elements of $S_H$, the symmetric group on $H$, and the monodromies of $X^{\mathrm{alg}}\to B_N^{\mathrm{alg}}$ as elements of $S_{H/C_k}$. The natural projection map between the monodromies is now just induced by the natural map on symmetric groups 
$$S_H\to S_{H/\mathcal{C}_k}.$$

Finally, it is clear from this point of view that the monodromy group of $X^{\mathrm{alg}}\to B_N^{\mathrm{alg}}$ is generated by $N-k-1$ elements. Indeed, we see at once that the monodromy group for $\pi$ is generated by $\rho_1,\dots,\rho_{N-1}$. By design, the map to the monodromy group of $X^{\mathrm{alg}}\to B_N^{\mathrm{alg}}$ is surjective. The quotienting by $\mathcal{C}_k$ will introduce an additional $k$ relations among the $\rho_i$, so that the group is generated by $N-k-1$ elements.

\subsection{Field of Definition}
As noted at the beginning of the section, Belyi curves are always definable over $\overline{\mathbb{Q}}$. The model that we provided for $X_N^{\mathrm{alg}}$ was visibly defined over $\mathbb{Q}(\zeta)$, where $\zeta$ is a primitive $2N$-th root of unity. In fact, the M\"{o}bius transformation that was used to map the $N$-th roots of $-1$ to the branch locus of $X_N^{\mathrm{alg}}$ necessarily maps each root of unity to the real line, so that the model is in fact defined over the maximal real subfield $\mathbb{Q}(\zeta)^+\subseteq\mathbb{Q}(\zeta)$. It turns out that we can do much better than this. We will show that  $X_N$ is definable over $\mathbb{Q}$. In order to do this, we will need to introduce some preliminaries. 

If $f\colon X\to Y$ is a morphism of varieties defined over $L$ and $\sigma$ is a field automorphism, then $f^\sigma\colon X^\sigma\to Y^\sigma$ is the map defined by $\sigma f \sigma^{-1}$. It fits into the following commutative diagram:
$$\xymatrix{X\ar_\sigma[d]\ar^f[r]&Y\ar^\sigma[d]\\
X^\sigma\ar^{f^\sigma}[r]&Y^\sigma}$$

\begin{definition}
Let $X$ be an algebraic variety defined over a number field $L$, and let $K\subseteq L$ be a subfield over which $L$ is Galois with Galois group $\Gamma$. A \emph{Galois descent datum} for $X$ with respect to $L/K$ is a family of birational isomorphisms $\{f_{\sigma}\colon X\to X^\sigma\}_{\sigma\in\Gamma}$ that satisfy the following cocycle condition:
$$f_{\sigma\tau}=f_\tau^\sigma f_\sigma$$
for every $\sigma,\tau\in\Gamma$.
\end{definition}

According to Theorem 1 of \cite{Hidalgo}, $X$ is definable over the smaller field $K$ if and only if $X$ admits a Galois descent datum with respect to $L/K$. This is a very strong condition: even if $X$ and $X^\sigma$ are birationally equivalent for each $\sigma\in\Gamma$, it may not be the case that $X$ is defined over $K$ if the birational maps are not compatible. It turns out that the additional structure gained from the fact that $X_N$ factors through $B_N$, together with the fact that we can translate the Galois action on $B_N$ into a geometric action via M\"{o}bius transformations, allows us to build the morphisms necessary for a Galois descent datum. 

Let $\zeta\in\overline{\mathbb{Q}}$ be a fixed primitive $N$-th root of unity, and let $B=\mathbb{P}^1(\mathbb{Q}(\zeta))$. Now set
$$f(z)=\frac{z-\zeta}{z-\zeta^{-1}}\xi$$ where $\xi=\frac{\zeta^3-\zeta^{-1}}{\zeta-\zeta^3}$. Then $f$ is a biregular isomorphism from $B$ to itself that maps the $N$-th roots of $-1$ into $\mathbb{Q}(\zeta)^+$ with $\zeta,\zeta^3,\zeta^{2N-1}$ mapping to $0,-1,\infty$ respectively. In order to make the following exposition clearer, it will be necessary to modify the notation used for the algebraic model of $X_N$. First, let us fix an ordering of the $N$-th roots of unity by setting
$$\zeta_i=\zeta^{2i-1},\ i=1,\dots, N.$$
Similarly, let us set
$$\mu_i=f(\zeta_i)$$
In particular, $\mu_1=0,\mu_2=-1,$ and $\mu_N=\infty$. If we let $X=C(-\mu_2,\dots,-\mu_{N-1})$ be the complete intersection of quadrics defined earlier, then $X$ is a model for $X_N$ defined over $\mathbb{Q}(\zeta)$. Then, by design, there is a map $\pi\colon X\to f(B)$ with branch locus $\{\mu_1,\dots,\mu_N\}$.

If we set $\zeta_0=f(0)$ and $\zeta_\infty=f(\infty)$, then the white vertices are the points of $\pi^{-1}(\zeta_0)$, the black vertices are the points of $\pi^{-1}(\zeta_\infty)$, and the faces of $X$ are the points of $\pi^{-1}(\mu_i)$ for $i=1,\dots, N$. From the description of $X$ as an intersection of quadrics, it is clear that all of these points have coordinates in $\overline{\mathbb{Q}}$. 

We will prove that $X$ has a model defined over $\mathbb{Q}$ by finding a Galois descent datum for an appropriate field. The first step in doing so is to show that $X$ and $X^\sigma$ are birationally equivalent for any field automorphism of $\overline{\mathbb{Q}}$. 

The idea is as follows. We know that $X^\sigma=C(-\mu_2^\sigma,\dots,-\mu_{N-1}^\sigma)$. On the other hand, if we set
$$f_\sigma(z)=f^\sigma f^{-1}(z),$$
then $f_\sigma$ is a biregular morphism from $f(B)$ to $f^\sigma(B)$ such that

$$\{\mu_1,\dots,\mu_N\}^\sigma=\{\mu_1^\sigma,\dots,\mu_N^\sigma\}.$$
Therefore, $X^\sigma$ is a model for $X_N$ that corresponds to choosing a different biregular morphism that identifies the $N$-th roots of $-1$ with an appropriate branch locus than $f$. By \cite{Carocca}, this  shows that they must be equivalent as Riemann surfaces. By making the isomorphism explicit, we will show that it can be used to build a Galois descent datum. 

\begin{proposition}\label{biregprop}
Let $N$ be the Galois closure (over $\mathbb{Q}$) of the smallest field containing $\mathbb{Q}(\zeta)$, and the coordinates of the faces of $X$. Let $\sigma \in\Gamma=\Gal(N/\mathbb{Q})$. There exists a biregular morphism $\eta\colon X\to X^\sigma$, defined over $N$, that fits into the following commutative diagram:

$$\xymatrix{X\ar^\eta[r]\ar_\pi[d]&X^\sigma\ar^{\pi_\sigma}[d]\\
B\ar^{f_\sigma}[r]&B_\sigma}$$
\end{proposition}

The proof of this proposition is located in the Appendix. We now use this result to prove the following theorem.

\begin{theorem}
\label{theorem:overQ}
The Riemann surface $X_N$ is defined over $\mathbb{Q}$. More precisely, let $M$ denote the Galois closure (over $\mathbb{Q}$) of the smallest field containing $\mathbb{Q}(\zeta)$,  the coordinates of a fixed white vertex $w\in\pi^{-1}(\zeta_0)$, and the coordinates of all of the faces of $X_N$. Then there exists a Galois descent datum for $X$ with respect to $M/\mathbb{Q}$. 
\end{theorem}

\begin{proof}
Fix $\sigma\in\Gamma$, where $\Gamma=\Gal(M/\mathbb{Q})$. According to Proposition \ref{biregprop}, there exists a biregular morphism $\eta\colon X\to X^\sigma$, defined over $N\subseteq M$, that fits into the appropriate commutative diagram. Using the deck transformation group $H\cong (\mathbb{Z}/2\mathbb{Z})^{N-1}$ of $\pi_\sigma$, we see that there are precisely $2^{N-1}$ such morphisms, all obtained by applying deck transformations to $\eta$. Both $\eta(w)$ and $w^\sigma$ are elements of $\pi_\sigma^{-1}(\zeta_0^\sigma)$. Since $H$ acts simply transitively on this fiber, there is a unique deck transformation $h_\sigma\in H$ such that $h_\sigma\eta(w)=w^\sigma$. That is, there is a unique biregular morphism that satisfies the appropriate commutative diagram and takes $w$ to $w^\sigma$; we will call this morphism $\eta_{\sigma}$. 

The collection $\{\eta_\sigma\}_{\sigma\in\Gamma}$ is a Galois descent datum. Indeed, one can check that
$$f_{\sigma\tau}=f_\tau^\sigma f_\sigma,$$ from which it follows that $\eta_{\sigma\tau}$ and $\eta_{\tau}^\sigma \eta_\sigma$ satisfy the same commutative diagram. Since both maps take $w$ to $w^{\sigma\tau}$, the maps must be equal by uniqueness. 
\end{proof}

Theorem \ref{theorem:overQ} implies that the curve $X_N$ can be defined over $\mathbb{Q}$. There is in \cite{Hidalgo} a formulation of a Galois descent datum for a group of automorphisms on a variety $X$, and a corresponding theorem  that says that the existence of a Galois descent datum implies that $X$ and the group of automorphisms are defined over a smaller field. The consequence for our situation is that each automorphism of the rational model for $X_N$ is defined over $\mathbb{Q}$. Indeed, $\{\eta_\sigma\}$ forms a Galois descent datum, and each automorphism in $H$ is already defined over $\mathbb{Q}$.

\begin{proposition}
The curves in  $\mathcal{X}_{(N,k)}$ are  definable over $\mathbb{Q}$. 
\end{proposition}
\begin{proof}
By the above discussion, $X_N$ and its automorphism group $H$ are defined over $\mathbb{Q}$. This readily implies that the quotients by fixed-point free subgroups will also be defined over $\mathbb{Q}$. 
\end{proof}

An interesting consequence of this result is that for any of the Fuchsian groups $\Gamma_{(N,k)}$ that uniformized the curves in $\mathcal{X}_{(N,k)}$,  $\Delta/\Gamma_{(N,k)}$ is defined over $\mathbb{Q}$. This provides us with an interesting family of ``modular'' curves that have integral models. 

\section{Conclusion}

Let us summarize the results of this paper. We have shown that every Adinkra chromotopology canonically defines a Riemann surface as a covering space over $\bP^1(\bC)$ branched over $\{0,1,\infty\}$. The study of Adinkras is of interest because they are graphical presentations of off-shell representations of the $(1|N)$ superalgebra. Describing such representations as Riemann surfaces provides a new approach to unanswered problems in supersymmetry. Many of the structures of an Adinkra are described in terms of the $2$-colored loops. This makes their description in terms of surfaces very natural, as statements about $2$-colored loops become properties of the $2$-cells in the natural CW-decomposition. We have given a complete description of these surfaces in multiple forms. The different descriptions of the surfaces associated to Adinkra chromotopologies, again, provide varied approaches to solving problems of interest. We list some of the salient features of the different descriptions here.
\begin{enumerate}
	\item The description of these Riemann surfaces as covering spaces of $\bP^1(\bC)$ allowed us to illustrate the relationship between Adinkra chromotopologies and quotients of the Hamming cube using the Galois theory of covering spaces.
	\item Covering space theory allowed us to give a Fuchsian uniformization of the surfaces in terms of torsion-free normal subgroups of the $(N,N,2)$-triangle group.
	\item Finally, we gave an explicit algebraic description of the surfaces as complete intersections of quadrics in projective space.
	\item The algebraic description allowed us to see the quotienting of Adinkra chromotopologies as reflections on affine coordinates in projective space.
	\item Properties and results about Adinkras can now be recast geometrically. We gave geometric interpretations of some of the important features of Adinkras.
\begin{enumerate}
\item We have shown that Adinkra chromotoplogies related by $R$-symmetry describe isomorphic Riemann surfaces.
\item We gave a description of the tensor product of Adinkras in terms of a multi-point connected sum of the associated surfaces.
\end{enumerate}
\end{enumerate}

An Adinkra is defined through stripping the supersymmetry algebra of its spatial dimensions. After showing that Adinkras naturally gives rise to very special Riemann surfaces, it is reasonable to ask whether this emergent ``extra'' dimension is physically meaningful (e.g., are there settings in which the Riemann surface can be naturally identified with a string worldsheet?).  Setting this sort of speculation aside, there is a completely natural way to ``remove'' this emergent geometric dimension; we may integrate along $1$-cycles to produce periods.  In other words, we can study the Jacobian abelian variety of our geometrized Adinkra.

Studying so-called regular dessins, Wolfart was able to show that there is a connection between the Jacobians of Riemann surfaces of the sort we have produced with the Jacobians of Fermat curves \cite{Wolfart}. More precisely, it is shown that any $1$-dimensional factor of the natural representation of the automorphism group of the dessin on its Jacobian corresponds to a complex multiplication factor which appears in the Jacobian of a Fermat curve. A complete analysis of the arithmetic properties of the Jacobians of our geometrized Adinkra chromotopologies constitutes work in progress.

The geometric interpretation of Adinkra chromotopologies ignores two additional structures an Adinkra possesses that must be included in order to give a full geometric description of off-shell representations of the $(1|N)$ superalgebra: an odd dashing and height assignment. In a subsequent paper, \cite{Doran:2015}, we show that the odd dashing defines a spin structure on the associated Riemann surface. Following work by Donagi and Witten \cite{Donagi:2013}, the addition of a spin structure allows us to define a canonical super Riemann surface structure with Ramond punctures. We also show in \cite{Doran:2015} that the Adinkra height assignments define a discrete Morse function on the super Riemann surface in the sense of both Banchoff \cite{Banchoff:1970} and Forman \cite{Forman:1998a,Forman:1998b}. The height assignment simultaneously admits an interpretation as a divisor on the (super) Riemann surface. Operations such as raising and lowering of nodes \cite{Gates:1995}, which play a key role in the physical application of Adinkras, are now geometrically meaningful operations on these ``Morse divisors".

Geometrized chromotopologies are very special as Riemann surfaces, and they remain so even when viewed as Belyi curves. While not the focus of this current paper, spin structures that correspond to odd dashings on an Adinkra are likewise distinguished. The same is true for the Morse divisors coming from height assignments.  The fact that Adinkras correspond to very special points in a moduli space of well-studied geometric objects may provide a key new tool for understanding supersymmetric representation theory.  It is our hope that the category of spin curves with Morse divisor, which has emerged through geometrizing Adinkrizable supermultiplets, may also be naturally broadened to include geometric incarnations of both non-Adinkrizable supermultiplets \cite{Hubsch:2013,Doran:2013}, and worldline reductions (``shadows'') of on-shell supermultiplets of physical interest.

\section*{Appendix: Proofs of Propositions \ref{RotProp} and \ref{biregprop}}
\begin{proof}[Proof of Proposition \ref{RotProp}]
We now supply the details for the proof of Proposition \ref{RotProp}. Let $\rho$ denote the automorphism of $\mathbb{P}^{N-1}(\mathbb{C})$ given by 
\begin{eqnarray}
\rho[x_1\colon x_2\colon\cdots\colon x_N]&=&[x_N\colon c_2x_1\colon\cdots\colon c_N x_{N-1}],\nonumber
\end{eqnarray}
where $c_2=\sqrt{-\mu_{N-1}}$, $c_3=1$, and $c_i=\sqrt{\mu_{i-1}}$ for $i=4,\dots,N$. For brevity, let $x=[x_1,\dots, x_N]$; we need to verify that $\rho(x)$ satisfies the $N-3$ equations that cut out $X_N^{alg}$. Plugging $\rho(x)$ into the first defining equation, we are reduced to showing that

\begin{equation}\label{rhocheck}x_N^2-\mu_{N-1}x_1^2+x_2^2=0\end{equation}
But this equation holds, being the defining equation of $X_N^{alg}$ that corresponds to the the variable $x_N$. 

The other defining equations of $X_N$ are
$$-\mu_{i}x_1^2+x_2^2+x_{i+1}^2=0,\ i=2,\dots,N-1.$$
Plugging in $\rho(x)$, we must therefore verify that
$$-\mu_i x_N^2-\mu_{N-1}x_1^2+\mu_i x_i^2=0$$ for $3\leq i\leq N-1$. Using \eqref{rhocheck} to replace $x_N^2$ with $\mu_{N-1}x_1^2-x_2^2$, we are reduced to checking that
\begin{equation}\label{35a}
\mu_{N-1}\frac{-\mu_i-1}{\mu_i}x_1^2+x_2^2+x_i^2=0.
\end{equation}
We claim that 
\begin{equation}\label{35a}\mu_{N-1}\frac{-\mu_i-1}{\mu_i}=-\mu_{i-1}.\end{equation}
This will complete the proof of our claim about $\rho$, since equation \eqref{35a} will then be true, being the same as the following defining equation of $X_N^\mathrm{alg}$:
$$ -\mu_{i-1} x_1^2 + x_2^2 + x_i^2 = 0. $$

Recall that the cross-ratio of any four distinct elements of a field is defined as
$$(z_1,z_2;z_3,z_4)=\frac{z_1-z_3}{z_1-z_4}\cdot\frac{z_2-z_4}{z_2-z_3}.$$
In terms of cross-ratios, 
$$-\mu_i=(\zeta_i,\zeta_1;\zeta_3,\zeta_N)$$ for all $1\leq i\leq N$. On the other hand, the action of the symmetric group on  cross-ratio is well known. In particular, we have
$$\frac{-\mu_i-1}{-\mu_i}=(\zeta_i,\zeta_N;\zeta_2,\zeta_1),$$

Writing \eqref{35a} in terms of cross-ratios and canceling out terms on both sides,  we are left with checking that the following equality holds:
$$\frac{\zeta_{N-1}-\zeta_1}{\zeta_{N-1}-\zeta_2}\cdot\frac{\zeta_2-\zeta_N}{\zeta_2-\zeta_1}\cdot\frac{\zeta_i-\zeta_2}{\zeta_i-\zeta_1}\cdot\frac{\zeta_N-\zeta_1}{\zeta_N-\zeta_2}=\frac{\zeta_{i-1}-\zeta_1}{\zeta_{i-1}-\zeta_2}\cdot\frac{\zeta_2-\zeta_N}{\zeta_2-\zeta_1}.$$
Observe that  each side in the above expression is a positive real number. Therefore, it suffices to show that both sides have the same absolute value. Using the fact that the $N$-th roots of $-1$ form a regular $N$-gon inscribed in the unit-circle, it can be checked that the two sides of the equation are indeed equal. For example, we see that
$$|\zeta_{i-1}-\zeta_1|=|\zeta_{i}-\zeta_2|,$$
so we can cancel out the corresponding quantities on each side. We have proved our claim about $\rho$.

Finally, one can check that the rotation on $B_N^{alg}$ is given by
$$r(z)=\frac{\mu_{N-1}}{z-\mu_{N-1}}.$$
Therefore, 
$$r(\pi(x))=-\frac{\mu_{N-1}x_1^2}{x_N^2}=\pi(\rho(x)),$$
so that $\rho$ is indeed a deck transformation of the Belyi map. 
\end{proof}
 
\begin{proof}[Proof of Proposition \ref{biregprop}]

We will describe the map $\eta$ explicitly as the composition of a permutation of the coordinates followed by a diagonal morphism on the ambient space.
The action of $\sigma$ on $\mathbb{Q}(\zeta)$ can be encoded by a permutation of the set $\{1,\dots,N\}$ determined by
$$\zeta_{\sigma(i)}=\zeta_i^\sigma.$$
Moreover, 
we have 
\begin{equation}\label{GalEq}
f_\sigma (\mu_{\sigma(i)})=\mu_i^\sigma
\end{equation}

Notice that the face corresponding to $\pi^{-1}(\mu_i)$ is given by $x_{i+1}=0$, where we compute $i+1$ modulo $N$ using $N$ as our representative for $0$.

We will consider three cases. First,  suppose that $\sigma|_{\mathbb{Q}(\zeta)}=\mathrm{id}$. Then $X^\sigma=X$, 
$f_\sigma=\mathrm{id}$, and it suffices to take $\eta=\mathrm{id}$.

Next, let us suppose that $\sigma|_{\mathbb{Q}(\zeta)}\colon\zeta\mapsto\zeta^{-1}$. Note that in this case, we still have $X^\sigma=X$, but since $f_\sigma$ is non-trivial, we cannot simply take a deck transformation for the requisite diagram to commute. As a permutation of the set $\{1,\dots,N\}$, $\sigma$ is the product of the transpositions $(i,N-i+1)$ for $1\leq i\leq \lfloor N/2\rfloor$. Let $s$ and $d$ be the automorphisms of $\mathbb{P}^{N-1}(\mathbb{C})$ defined by

\begin{eqnarray}s\colon\mathbb{P}^{N-1}(\mathbb{C})&\to &\mathbb{P}^{N-1}(\mathbb{C})\nonumber\\
\left[x_1\colon x_2\colon\cdots\colon x_N\right]&\mapsto &\left[x_{\sigma(N)+1}\colon x_{\sigma(1)+1}\colon x_{\sigma(2)+1}\colon\cdots \colon x_{\sigma(N-1)+1}\right]\nonumber\\
&=&\left[x_2\colon x_1\colon x_{N-3}\colon\cdots\colon x_3\right]\nonumber\\
d\colon\mathbb{P}^{N-1}(\mathbb{C})&\to &\mathbb{P}^{N-1}(\mathbb{C})\nonumber\\
\left[x_1\colon x_2\colon\cdots\colon x_N\right]&\mapsto &\left[x_1\colon c_2 x_2\colon\cdots\colon c_Nx_N\right],\nonumber
\end{eqnarray}

where $c_3=1$, $c_i=\sqrt{-\mu_{i-1}}$ for $4\leq i\leq N$, and $c_2=c_N$. Here, the notation $\sqrt{-\mu_{i-1}}$ means either choice of square root of $-\mu_{i-1}$ in $\overline{\mathbb{Q}}$.  Then $\eta=d\circ s$ restricts to an isomorphism from $X$ to $X^\sigma$ that makes the requisite diagram commute. The details for why $\eta$ respects the equations for $X$ and $X^\sigma$ are contained in the lemma below. One can check that the face corresponding to $\mu_1$ is given by
$$\pi^{-1}(0)=[1\colon 0\colon\pm\sqrt{-1}\colon\pm\sqrt{\mu_3}\colon\cdots\colon\pm\sqrt{\mu_N}],$$
from which it is clear that $\eta$ will be defined over $M$. 

Lastly, suppose that $\sigma$ restricts to some other automorphism of $\mathbb{Q}(\zeta)$. Consider the face of $X$ given by
$$\pi^{-1}(\mu_{\sigma(1)})=[a_1\colon a_2\colon\cdots\colon a_N],$$ 
where the $a_i$ are given as follows:

 $$a_i=\left\{\begin{array}{lll}1&\textrm{if }& i=1\\
\sqrt{\mu_{\sigma(1)}}&\textrm{if }& i=2\\
\sqrt{\mu_{i-1}-\mu_{\sigma(1)}}&\textrm{if }&i\geq 3.

\end{array}\right.$$

Notice that  by \eqref{GalEq}, this  face of $X$  maps to a face of $X^\sigma$ corresponding to $0$.  
Let $s$ and $d$ be the automorphisms of $\mathbb{P}^{N-1}(\mathbb{C})$ given by
\begin{eqnarray}s\colon\mathbb{P}^{N-1}(\mathbb{C})&\to &\mathbb{P}^{N-1}(\mathbb{C})\nonumber\\
\left[x_1\colon x_2\colon\cdots\colon x_N\right]&\mapsto &\left[x_{\sigma(N)+1}\colon x_{\sigma(1)+1}\colon x_{\sigma(2)+1}\colon\cdots \colon x_{\sigma(N-1)+1}\right]\nonumber\\
d\colon\mathbb{P}^{N-1}(\mathbb{C})&\to &\mathbb{P}^{N-1}(\mathbb{C})\nonumber\\
\left[x_1\colon x_2\colon\cdots\colon x_N\right]&\mapsto &\left[x_1\colon c_2 x_2\colon\cdots\colon c_Nx_N\right],\nonumber
\end{eqnarray}
where
\begin{eqnarray}
c_2&=& \sqrt{\mu_{\sigma^{-1}(N)}^\sigma}\nonumber\\
c_i&=&\sqrt{\mu_{i-1}^\sigma}\cdot\frac{a_{\sigma(N)+1}}{a_{\sigma(i-1)+1}},\ i\geq 3.\nonumber
\end{eqnarray}
Then $\eta=d\circ s$ restricts to an isomorphism from $X$ to $X^\sigma$ that makes the requisite diagram commute. Once more, the details for why $\eta$ respects the equations for $X$ and $X^\sigma$ are contained in the lemma below.  Since the $c_i$ are defined in terms of the coordinates of a face, it follows that $\eta$ will be defined over $M$. 

\end{proof}
\begin{lemma}\label{Lemma}
The morphisms $\eta$ defined above restrict to morphisms $X\to X^\sigma$.
\end{lemma}
\begin{proof}
We  start with the second case above (the first being trivial), where we claimed that we could take 
$$\eta=[x_2\colon c_2 x_1\colon c_3x_{N-3}\cdots\colon c_Nx_3]$$
with $c_3=1$, $c_i=\sqrt{-\mu_{i-1}}$ for $4\leq i\leq N$, and $c_2=c_N$. Before showing that $\eta$ defines an isomorphism, we remark that 
\begin{equation}\label{c2transp}-c_2^2=c_i^2\mu_{N-i+2},\ i\geq 3.\end{equation}
The proof  uses the same cross-ratio trick we employed earlier and is omitted. 

The verification that $\eta$ is an isomorphism boils down to verifying that
\begin{equation}\label{37a}
-\mu_{i-1}x_2^2+c_2^2x_1^2+c_i^2x_{N-i+3}^2=0
\end{equation}
for all $i\geq 3$. Dividing through by $c_i^2 = -\mu_{i-1}$ and using equation \eqref{c2transp}, we see that \eqref{37a} has the equivalent form
$$-\mu_{N-i+2}x_1^2+x_2^2+x_{N-i+3}^2=0$$
and this equation holds because it is a defining equation for $X_N^{alg}$.

For the remaining case, we again start by showing that $c_2^2$ can be expressed in terms of each $\mu_i$ for $i\geq 3$. In fact, with  notation as in Proposition \ref{biregprop}, we have
$$c_2^2=\mu_{i-1}^\sigma\left(\frac{b_{\sigma(i-1)+1}}{a_{\sigma(i-1)+1}}\right)^2=\mu_{i-1}^\sigma\frac{\mu_{\sigma(N)}-\mu_{\sigma(i-1)}}{\mu_{\sigma(1)}-\mu_{\sigma(i-1)}}.$$
To see why, note that if $i-1=\sigma^{-1}(N)$, then the above just says
$$c_2^2=\mu_{\sigma^{-1}(N)}^\sigma=f_\sigma(\mu_N),$$
which is true just by how we defined $c_2$. 
We will now show that for all other $i\geq 3$ for which $i-1\neq\sigma^{-1}(N)$, we have
$$\mu_{i-1}^\sigma\cdot\frac{\mu_{\sigma(N)}-\mu_{\sigma(i-1)}}{\mu_{\sigma(1)}-\mu_{\sigma(i-1)}}=f_\sigma(\mu_N).$$
We know that $\mu_{i-1}^\sigma=f^\sigma(\zeta_{\sigma(i-1)})$ and $f_\sigma(\mu_N)=f^\sigma(\zeta_N)$.
After canceling the factor of $\xi^\sigma$ coming from both of these identities, we are reduced to showing that
$$\frac{\zeta_{\sigma(i-1)}-\zeta_{\sigma(1)}}{\zeta_{\sigma(i-1)}-\zeta_{\sigma(N)}}\cdot\frac{\mu_{\sigma(N)}-\mu_{\sigma(i-1)}}{\mu_{\sigma(1)}-\mu_{\sigma(i-1)}}=\frac{\zeta_N-\zeta_{\sigma(1)}}{\zeta_N-\zeta_{\sigma(N)}}.$$

In terms of cross-ratios,
$$\frac{\mu_{\sigma(N)}-\mu_{\sigma(i-1)}}{\mu_{\sigma(1)}-\mu_{\sigma(i-1)}}=(\mu_{\sigma(N)},\mu_{\sigma(1)};\mu_{\sigma(i-1)},\infty).$$
Since M\"{o}bius transformations preserve cross-ratios, the above cross-ratio is equal to
$$(\zeta_{\sigma(N)},\zeta_{\sigma(1)};\zeta_{\sigma(i-1)},\zeta_N)=\frac{\zeta_{\sigma(N)}-\zeta_{\sigma(i-1)}}{\zeta_{\sigma(1)}-\zeta_{\sigma(i-1)}}\cdot\frac{\zeta_{\sigma(1)}-\zeta_N}{\zeta_{\sigma(N)}-\zeta_N}.$$

Therefore, we are left to verify 
$$\frac{\zeta_{\sigma(i-1)}-\zeta_{\sigma(1)}}{\zeta_{\sigma(i-1)}-\zeta_{\sigma(N)}}\cdot
\frac{\zeta_{\sigma(N)}-\zeta_{\sigma(i-1)}}{\zeta_{\sigma(1)}-\zeta_{\sigma(i-1)}}\cdot\frac{\zeta_{\sigma(1)}-\zeta_N}{\zeta_{\sigma(N)}-\zeta_N}=
\frac{\zeta_N-\zeta_{\sigma(1)}}{\zeta_N-\zeta_{\sigma(N)}}.$$
We immediately see that everything cancels out nicely, so that 

$$c_2^2=\mu_{i-1}^\sigma\cdot\frac{\mu_{\sigma(i-1)}-\mu_{\sigma(N)}}{\mu_{\sigma(i-1)}-\mu_{\sigma(1)}},$$
as desired.

We now show that these choices of $c_i$ give rise to an isomorphism $\eta$. 
This amounts to verifying 
\begin{equation}\label{etacheck}-\mu_{i-1}^\sigma x_{\sigma(N)+1}^2+c_2^2x_{\sigma(1)+1}^2+c_{i}^2x_{\sigma(i-1)+1}^2=0,\ i\geq 3.
\end{equation}

There are three cases. First, suppose that $i-1=\sigma^{-1}(N)$, so that $x_{\sigma(i-1)+1}=x_1$. In this case, 
$$c_{i}^2=\mu_{i-1}^\sigma(\mu_{\sigma(N)}-\mu_{\sigma(1)})$$ and
$$c_2^2=\mu_{i-1}^\sigma.$$

Using the equations for $X_N$, we may use the identity
\begin{equation}\label{38a}
x_{\sigma(N)+1}^2=\mu_{\sigma(N)}x_1^2-x_2^2
\end{equation}
in equation \eqref{etacheck} and divide by $c_2^2$, at which point it suffices to show that
$$x_{\sigma(1)+1}^2+x_2^2-\mu_{\sigma(1)}x_1^2=0.$$
But this is exactly the defining equation for $X_N$ corresponding to $x_{\sigma(1)+1}$, so it is indeed true. 

The next case is  $\sigma(i-1)=1$, so that $x_{\sigma(i-1)+1}=x_2$. In this case, we have $$c_i^2=\mu_{i-1}^\sigma\frac{\mu_{\sigma(N)}-\mu_{\sigma(1)}}{\mu_{\sigma(1)}}$$ and
$$c_2^2=\mu_{i-1}^\sigma\frac{\mu_{\sigma(N)}}{\mu_{\sigma(1)}}.$$
Again using the identity in \eqref{38a}, we are left with verifying that
$$x_{\sigma(1)+1}^2+x_2^2-\mu_{\sigma(1)}x_1^2=0,$$ which is valid because, as before, it is one of the defining equations for $X_N$.  

Finally, for all other values of $i$, we use both of the identities
\begin{eqnarray}
x_{\sigma(N)+1}^2&=&\mu_{\sigma(N)}x_1^2-x_2^2,\nonumber\\
x_{\sigma(1)+1}^2&=&\mu_{\sigma(1)}x_1^2-x_2^2.\nonumber
\end{eqnarray}
Plugging them into \eqref{etacheck} and arguing similarly as above, we are reduced to checking that
$$\frac{(-\mu_{i-1}^\sigma\mu_{\sigma(N)}+c_2^2\mu_{\sigma(1)})}{c_i^2}=-\mu_{\sigma(i-1)},$$
which follows easily  from our work above. 
\end{proof}

\section*{Acknowledgements}

We would like to thank S.J. Gates, Jr., and T. H\"{u}bsch for extended discussions and useful suggestions while writing this paper. We would also like to thank J. Iverson for computing the generators of the monodromy groups for some of the examples in Sage. CD and SMD acknowledge the support from the Natural Sciences and Engineering Resource Council of Canada, the Pacific Institute for the Mathematical Sciences, and the McCalla Professorship at the University of Alberta. GL acknowledges the support by a grant from the Simons Foundation (Award Number 245784).

\bibliographystyle{hplain}

\address[Charles Doran, Jordan Kostiuk]{Department of Mathematical and Statistical Sciences\\
University of Alberta\\
Edmonton, AB T6G 2G1, Canada}
\email{doran@math.ualberta.ca, jkostiuk@ualberta.ca}
  
\author[K. Iga]{Kevin Iga, test}
\address[Kevin Iga]{Natural Science Division\\
Pepperdine University\\
Malibu, CA 90263, USA} 
\email{kiga@pepperdine.edu}

\author[G. Landweber]{Greg Landweber}
\address[Greg Landweber]{Mathematics Program\\
Bard College\\
Annandale-on-Hudson, NY 12504-5000, USA} 
\email{gregland@bard.edu}

\author[S. M\'{e}ndez-Diez]{Stefan M\'{e}ndez-Diez}
\address[Stefan M\'{e}ndez-Diez]{Department of Mathematics \& Statistics\\
Utah State University\\
Logan, UT 84322-3900, USA}
\email{stefan.md@usu.edu}

\end{document}